\documentclass[11pt,letterpaper]{article}
\usepackage[english]{babel}

\usepackage[margin=1in]{geometry}
\usepackage{authblk}

\usepackage[vlined,ruled,linesnumbered]{algorithm2e}
\SetKwInput{KwData}{Input}
\SetKwInput{KwResult}{Output}

\usepackage{amsthm}
\usepackage{amsmath}
\usepackage{amssymb}
\usepackage{graphicx}
\usepackage{thmtools} 
\usepackage{thm-restate}
\usepackage{enumerate}
\usepackage{bm}

\newcommand\numberthis{\addtocounter{equation}{1}\tag{\theequation}}
\newtheorem{theorem}{Theorem}[section]

\newtheorem{proposition}[theorem]{Proposition}
\newtheorem{lemma}[theorem]{Lemma}

\newtheorem{Assumption}{Assumption}

\theoremstyle{definition}
\newtheorem{Definition}[theorem]{Definition}
\newtheorem{example}[theorem]{Example}
\newtheorem{Remark}[theorem]{Remark}
\newtheorem{remark}[theorem]{Remark}

\usepackage{tikz}
\foreach \x in {A,...,Z}{%
    \expandafter\xdef\csname m\x\endcsname{\noexpand\mathbf{\x}}
}
\foreach \x in {a,...,t,v,w,x,y,z}{%
    \expandafter\xdef\csname m\x\endcsname{\noexpand\mathbf{\x}}
}
\newcommand{\0}{\mathbf{0}}

\newcommand{\pr}[2]{\left\langle #1, #2 \right\rangle}
\newcommand{\R}{\mathbb{R}}

\newcommand{\Rp}{\mathbb{R}_{\ge0}}

\newcommand{\Rpp}{\mathbb{R}_{>0}}
\newcommand{\defeq}{\stackrel{\scriptscriptstyle{\mathrm{def}}}{=}}
\newcommand{\supp}{\operatorname{supp}}

\newcommand{\subg}{\partial^{\vee}}
\newcommand{\supg}{\partial^{\wedge}}

\newcommand{\agents}{{\mathcal{A}}}
\newcommand{\dom}{\operatorname{dom}}
\newcommand{\din}{\partial^{\mathrm{in}}}
\newcommand{\dout}{\partial^{\mathrm{out}}}
\newcommand{\Primal}{\operatorname{Primal}}
\newcommand{\Dual}{\operatorname{Dual}}
\newcommand{\Genflow}{\operatorname{GenflowLP}}
\newcommand{\goods}{{\mathcal{M}}}
\newcommand{\net}[2]{\operatorname{net}_{#1}(#2)}
\newcommand{\demand}[1]{{\mathcal{D}_{#1}}}
\newcommand{\galedemand}[1]{\mathcal{GD}_{#1}}
\newcommand{\EGdual}{\varphi}
\SetKwFunction{Pullback}{Pullback-Flow}
\SetKwFunction{Reroute}{Reroute-Flow}
\newcommand{\Gaux}{G^{\textrm{aux}}}
\newcommand{\Eaux}{E^{\textrm{aux}}}
\newcommand{\rev}[1]{\overleftarrow{#1}}

\newcommand{\cp}{h}
\newcommand{\mcp}{\mh}
\newcommand{\FQ}[1]{\operatorname{pdom}(#1)}

\newcommand{\approxi}{\alpha}

\newcommand{\ee}{\mathrm{e}}

\usepackage{hyperref}
\usepackage{tikz-cd}

\usepackage{natbib}
\setcitestyle{authoryear,open={(},close={)}}
\renewcommand{\cite}[1]{\citep{#1}}

\title{Approximating Competitive Equilibrium by Nash Welfare\thanks{A preliminary version of this paper was presented at the \emph{2025 Annual ACM-SIAM Symposium on Discrete Algorithms (SODA)~\cite{GargTV25}}.
This work was supported by the European Research Council (ERC) under the European Union's Horizon 2020 research and innovation programme (grant agreement no.~ScaleOpt--757481). J.~Garg was supported by NSF Grants CCF-1942321 and CCF-2334461. Y. Tao was supported by Grant 2023110522 from SUFE, National Key R\&D Program of China (2023YFA1009500), NSFC grant 61932002, Innovation Program of Shanghai Municipal Education Commission, and Fundamental Research Funds for Central Universities. Part of the work was done while L.~V\'egh was affiliated with the London School of Economics and Political Science, UK. \\ The authors are grateful to the anonymous referees for their comments and suggestions that helped to improve the presentation of the paper.}}
\author[1]{Jugal Garg}
\author[2]{Yixin Tao} 
\author[3,4]{L{\'{a}}szl{\'{o}} A. V{\'{e}}gh}
\affil[1]{University of Illinois at Urbana-Champaign, USA}
\affil[2]{Key Laboratory of Interdisciplinary Research of Computation and Economics\\ Shanghai University of Finance and Economics, China }
\affil[3]{Hertz Chair for Algorithms and Optimization, University of Bonn, Germany}
\affil[4]{Corvinus Institute for Advanced Studies, Corvinus University, Budapest, Hungary} 
\date{}

\begin{document}
\maketitle
\thispagestyle{empty}
\begin{abstract}
We study the relationship between two central concepts in the allocation of divisible goods: competitive equilibrium (CE) and allocations that maximize Nash welfare, i.e., allocations where the weighted geometric mean of the utilities is maximal. When agents have homogeneous concave utility functions, these concepts coincide: the classical Eisenberg--Gale convex program that maximizes Nash welfare over feasible allocations yields a competitive equilibrium. However, they diverge for non--homogeneous utilities. From a computational perspective, maximizing Nash welfare amounts to solving a convex program for any concave utility functions, whereas computing CE becomes PPAD-hard already for separable piecewise linear concave (SPLC) utilities. 

We introduce the concept of \emph{Gale-substitute} utility functions, an analogue of the weak gross substitutes (WGS) property for the so-called Gale demand system. For Gale-substitutes utilities, we show that any allocation maximizing Nash welfare provides an approximate-CE with surprisingly strong guarantees, where every agent gets at least half the \emph{maximum} utility they can get at any CE, and is approximately envy-free.
Gale-substitutes include utility functions where computing CE is PPAD hard, such as all separable concave utilities and the previously studied non-separable class of Leontief-free utilities. We introduce a broad new class of utility functions called \emph{generalized network utilities} based on the generalized flow model. This class includes SPLC and Leontief-free utilities, and we show that all such utilities are Gale-substitutes.

Conversely, although some agents may get much higher utility at a Nash welfare maximizing allocation than at a CE, we show a `price of anarchy' type result: for general concave utilities, every CE achieves at least $(1/\ee)^{1/\ee} > 0.69$ fraction of the maximum Nash welfare, and this factor is tight. 
\end{abstract}

\newpage
\setcounter{page}{1}

\section{Introduction}
We investigate a setting where a set $\goods$ of divisible goods needs to be fairly allocated among a set $\agents$ of $n$ agents in a cardinal utility framework. Each agent $i\in \agents$ has preferences given by a utility function $u_i(.)\,:\, \Rp^\goods\to \Rp$. We make the standard assumptions that $u_i(.)$'s are monotone non-decreasing, concave, and satisfy $u_i(\0)=0$.\footnote{All concepts discussed in the Introduction will be formally defined in Section~\ref{sec:prelim}.} Two classical allocation concepts are Nash welfare and competitive equilibrium from equal incomes (CEEI). In Nash welfare, we select an allocation  $(\mx_i)_{i\in \agents}$ that maximizes $\prod_{i\in\agents} u_i(\mx_i)$. In CEEI, each agent is initially given an equal share of all goods and allowed to trade with others; the resulting allocation $(\mx_i)_{i\in \agents}$ corresponds to a competitive equilibrium (CE) outcome.

Both these concepts have desirable properties that have made them widely popular in fair division research. Nash welfare derives its name from its roots in cooperative bargaining, introduced by Nash~(\citeyear{nash1950bargaining}). In this model, a convex set $K\subseteq \R^n$ of feasible outcomes is given, along with a disagreement point $\md\in K$ and concave utility functions $u_i\,:\, K\to \Rp$. Nash showed that the unique bargaining solution satisfying natural axioms---Pareto optimality, invariance under affine transformations, symmetry, and independence of irrelevant alternatives---is the one that maximizes the geometric mean of the agents' gains over the disagreement point. Taking $K$ as the set of feasible allocations and $\md=\0$, the Nash bargaining solution exactly corresponds to the Nash welfare maximizing allocation in this setting (see also \cite{vazirani2012notion}). 

CEEI, introduced by \citet{Varian74}, also satisfies desirable properties such as Pareto optimality, scale invariance, and symmetry. In addition, it guarantees envy-freeness, which may not be the case for Nash welfare. For further discussion and references on these concepts, see \citet[Chapters 3 and 7]{Moulin03}. 

Nash welfare maximization and CEEI remarkably coincide for degree one positively homogeneous utilities, i.e., if $u(\alpha \mx)=\alpha u(\mx)$ for any $\alpha>0$.
This was first shown by \citet{eisenberg1959consensus} for linear utilities (i.e., $u(\mx) = \sum_{j\in\goods} v_{j}x_j$), 
and later extended by \citet{Eisenberg61} to the homogeneous case that includes classical examples of constant elasticity of substitution (CES), Cobb--Douglas, and Leontief utility functions. 

In general, Nash welfare maximization and CEEI can lead to different outcomes. Moreover, while the Nash welfare solution yields a unique utility profile for the agents, CEEI may admit multiple equilibria with significantly different utilities across agents. This raises a nontrivial \emph{equilibrium selection} problem~\cite{ArrowH71,harsanyi1998general}.

The problem of efficiently computing equilibria has been an important topic in economics, optimization, and computer science over the past decades. The existence of CEEI follows from classical results in general equilibrium theory by Arrow and Debreu~(\citeyear{arrow1954existence}) and by McKenzie~(\citeyear{Mckenzie1954}), based on fixed point theorems. While these results apply to more general exchange market models, they do not yield efficient computational methods. In fact, the computational complexity of finding equilibria is generally negative: even for the seemingly simple case of separable piecewise concave (SPLC) utilities, computing a CE falls into the class of PPAD-hard~\cite{ChenT09} problems, suggesting that an efficient algorithm is unlikely to exist. The hardness persists even under bounded linear utilities, where the utility from each good is capped, i.e., 
$u(\mx) = \sum_{j\in \goods} \min\{v_{j}x_j, \ell_{j}\}$~\citep{bei2019earning}.
Examples~\ref{eg:1} and~\ref{eg:non-thrifty} 
show that the set of CE can be non-convex and disconnected even under such preferences.  In contrast, Nash welfare maximization corresponds to solving a convex program and can therefore be computed efficiently for any concave valuations.

Besides homogeneous utilities, another important class where CEEI can be computed in polynomial time is \emph{weak gross substitutes (WGS)} utility functions.\footnote{A utility function satisfies the weak gross substitutes property when raising the price of one good does not decrease the demand for any other good.} Several fast algorithms have been designed for computing equilibria with WGS utilities~\cite{bei2019ascending,codenotti2005market,codenotti2005polynomial, garg2023auction}. While the WGS class includes many homogeneous utilities, there are also non-homogeneous examples where the Eisenberg--Gale solution may not correspond to a CE. Together, WGS and homogeneous utility functions cover the majority of known settings in which a CE can be computed efficiently.

In this paper, we study the relationship between two fundamental concepts in the allocation of divisible goods: competitive equilibrium and maximum Nash welfare allocations. We reveal a strong connection between these solution concepts that extend well beyond the classical homogeneous case.

While the classical fair division concepts are CEEI and Nash welfare maximizing allocations, our results apply more broadly to the \emph{Fisher market model} and weighted Nash welfare maximization. In the Fisher model, a fixed supply of goods $\goods$ is to be allocated among agents $\agents$, where each agent $i\in \agents$ has cardinal preferences given by a utility function $u_i(.)$. Each agent also has a fixed budget $b_i>0$. At competitive equilibrium (CE), prices $\mp\in\Rp^\goods$ of goods and fractional allocation $(\mx_i)_{i\in \agents}$ satisfy two conditions: {\em (i)} each agent $i$ gets an optimal bundle $\mx_i$ that maximizes $u_i(\mx_i)$ subject to the budget constraint $\langle \mp, \mx_i\rangle\le b_i$, and {\em (ii)} the market clears, meaning demand meets supply for every good. The CEEI setting corresponds to the special case where all budgets are equal, i.e., $b_i=1$ for all $i\in\agents$. Similarly, for given weights $b_i>0$, the weighted Nash welfare maximizing allocation is a feasible allocation of goods that maximizes $\prod_i u_i(\mx_i)^{b_i}$, or equivalently, $\sum_i b_i\log{u_i(\mx_i)}$. Classical results in \cite{eisenberg1959consensus,Eisenberg61} establish that these allocations coincide with CE outcomes for degree one positively homogeneous utilities and arbitrary budgets $b_i>0$. Henceforth, by Nash welfare maximizing allocation we refer to this weighted concept.

\subsection{Our contributions}
Our first result is a \emph{`price of anarchy'} type bound \cite{dubey,koutsoupias1999worst}.
One can interpret Nash welfare as a global measure of social utility. In contrast, a CE can be seen as an outcome reached by autonomous agents trying to improve their utilities individually. It turns out that, in the general setting of concave utilities, any CE recovers at least $(1/\ee)^{1/\ee} > 0.69$ fraction of the maximum Nash welfare, and this factor is tight (see Example~\ref{ex:poa}).

\begin{restatable}{theorem}{thmfour}\label{thm:4}
For Fisher markets under concave utility functions, the Nash welfare at any CE is at least $(1/\ee)^{1/\ee}$ times the maximum Nash welfare.
\end{restatable}

We then study guarantees of the maximum Nash welfare solution in the context of CE. As noted above, a major advantage of maximum Nash welfare is efficient computability: it can be arbitrarily well approximated for any concave utility functions, whereas computing CE becomes PPAD-hard already in very simple cases.
In this context, our first simple result shows that the maximum Nash welfare yields an approximate CE for all concave utility functions. As in the case of the Eisenberg--Gale program, one can interpret the optimal Lagrangian multipliers as prices.
By an $\alpha$-demand-approximate CE, we mean allocations and prices such that every agent receives at least $1/\alpha$-fraction of the maximum possible utility at those prices (see Definition~\ref{def:approx-market-eq}).  The result also implies approximate envy-freeness.

\begin{restatable}{theorem}{thmfirst}\label{thm:1}
For Fisher markets under concave utility functions, any Nash welfare maximizing allocation $(\my_i)_{i\in \agents}$, together with the dual prices $\mq$, form a 2-demand-approximate CE.
This also implies that $(\my_i)_{i\in \agents}$ is half envy-free when $b_i$'s are identical, i.e., $u_i(\my_i) \ge \tfrac12u_i(\my_k)$, for all $i,k\in\agents$. 
\end{restatable}

Theorem~\ref{thm:1} already shows that maximum Nash welfare provides remarkable fairness and efficiency guarantees for all concave utility functions. However,  there can be multiple equilibria, giving vastly different utilities to an agent. In Example~\ref{eg:2}, we show that there could be a gap of $\Omega(n)$ in the utility of some agent between two different competitive equilibria. Hence, some agent could be much better off in some CE than in the maximum Nash welfare solution. 

\medskip

As the main contribution of the paper, we show that maximum Nash welfare provides surprisingly strong guarantees for a broad class of utility functions. For this, we explore the \emph{Gale demand correspondence}~\cite{garg2023auction,nesterov2018computation}.
For an agent $i\in\agents$ with budget $b_i$ and utility function $u_i(.)$ at prices $\mq$, the Gale demand correspondence\footnote{In contrast, the standard demand correspondence comprises of the allocations $\mx_i$ that maximize $\{u_i(\mx_i)\ |\  \langle \mq, \mx_i\rangle \le b_i\,\}.$}
comprises of the allocations $\my_i$ that maximize 
\[b_i\log u_i(\my_i) - \langle \mq, \my_i\rangle\, .\]
 It follows from Lagrangian duality that in a Nash welfare maximizing solution, every agent receives an allocation from their Gale demand correspondence.

The reason behind the equivalence of CE and maximum Nash welfare for homogeneous utilities is that, for this class, the standard demand correspondence and the Gale demand correspondence coincide. Nevertheless, these two demands can behave very differently for non-homogeneous utility functions, e.g., the utility value from the standard demand is always monotone non-increasing, whereas the utility under Gale demand can be non-monotonic
(see Example~\ref{eg:mono}).

We identify a critical property of the Gale demand correspondence that we call \emph{Gale-substitutes property} (Definition~\ref{def:gale-subs}). This is analogous to the weak gross substitutes (WGS) property of the standard demand correspondence. In particular, it asserts that if some prices decrease then the demand for goods with unchanged prices may only decrease. 

Competitive equilibrium becomes more subtle when the utility functions are \emph{satiable}: there is an absolute maximum value of the utility that can be reached, regardless of the prices---bounded linear utilities mentioned above are such an example. To accommodate satiable utilities, we also introduce a stronger variant of the definition, called the \emph{$\Sigma$-Gale-substitutes property}. We note that for \emph{non-satiable} utilities, the Gale-substitutes and $\Sigma$-Gale-substitutes properties coincide. Figure~\ref{fig:containment} illustrates the containment relationships among various utility classes discussed in this paper.

\begin{figure}[t]
\includegraphics[width=8cm]{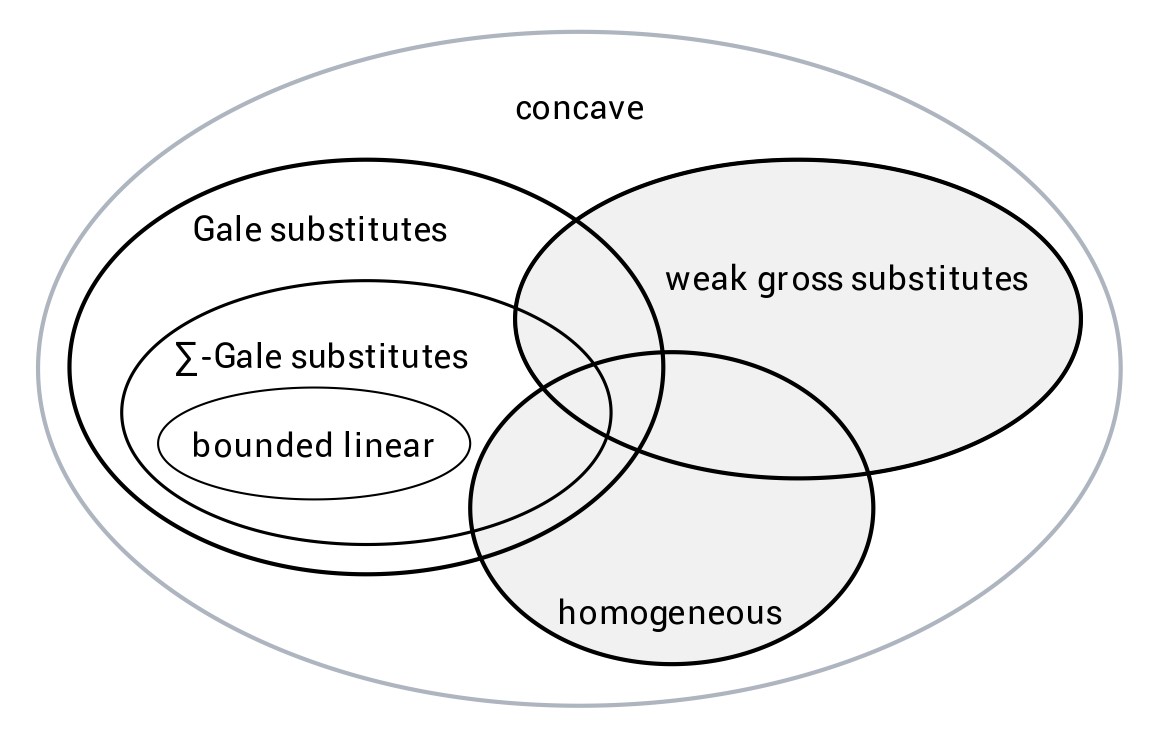}
\centering
\caption{The graph illustrates the containment relationships among different utility classes, with the classes where the Fisher market equilibrium can be computed efficiently highlighted in grey.}\label{fig:containment}
\end{figure}

Our main result shows that for this class of utility functions, the utility of any agent at the Nash welfare maximizing allocation is at least half the \emph{maximum} utility they may get at \emph{any} CE.

\begin{restatable}{theorem}{mainresult} \label{thm::main-result}
    Consider a Fisher market instance where all utility functions are  $\Sigma$-Gale-substitutes. 
    Then, for any competitive equilibrium $((\mx_i)_{i\in \agents}, \mathbf{p})$ and any  Nash welfare maximizing allocation $(\my_i)_{i\in \agents}$, $u_i(\my_i) \geq \frac{1}{2} u_i(\mx_i)$ holds for every $i\in\agents$.
\end{restatable}

As a counterpart to the \emph{`price of anarchy'} result in Theorem~\ref{thm:4}, one could call this a \emph{`price of authority'} result. The agents can either accept an allocation proposed by a central authority on the principle of maximizing Nash welfare or converge to a market solution (using some dynamics). By accepting the centrally proposed solution, each agent may be worse off by at most a factor 2 compared to the best possible equilibrium outcome. On the other hand, outcomes at different competitive equilibria can be very different, and the possible outcomes can be difficult to predict given the computational hardness of equilibria. Hence, for ($\Sigma$-)Gale-substitutes utility functions, the Nash welfare maximizing allocation provides a guaranteed good outcome to all agents simultaneously while also being efficiently computable. Thus, Theorem~\ref{thm::main-result} establishes a strong relationship between CE and maximum Nash welfare beyond homogeneous functions and provides a natural alternative to the \emph{equilibrium selection} problem~\cite{ArrowH71,harsanyi1998general} in this setting where multiple equilibria are possible.

We note that the statement of Theorem~\ref{thm::main-result} does not hold for arbitrary concave utilities. In Example~\ref{example:non-gale}, we show that without the Gale-substitutes property, agents may get much higher utilities in CE than in the maximum Nash welfare allocations.

\paragraph{Examples of Gale-substitute utilities}
It turns out that ($\Sigma$-)Gale-substitutes capture a wide range of natural utility functions.
We call a utility function $u(.)$ \emph{separable} if $u(\mx)=\sum_{j\in\goods}v_j(x_j)$. This class gives an important example of  $\Sigma$-Gale-substitutes. Note that such utilities, in general, are neither homogeneous nor WGS. As mentioned above, competitive equilibrium computation is already PPAD-hard for the separable piecewise linear concave (SPLC) case when each $v_j$ is piecewise linear. This class has attracted significant attention on its own \cite{AnariMGV18,ChenDDT09,ChenPY13,GargMSV15,VaziraniY11}.
\begin{restatable}{theorem}{sepgale}\label{thm:sep-gale}
Every separable utility function is $\Sigma$-Gale-substitutes.
\end{restatable}

We consider two further examples of non-separable, piecewise linear Gale-substitutes utilities: 
\begin{itemize}
\item \emph{Leontief-free (LF) utilities}~\cite{FHHH23,garg2014dichotomies,GoyalSG23}: 
Consider a set $K$ of segments, and let
$u(\mx) \defeq \max \sum_{k\in K} u_k$: $u_k = \sum_{j\in \goods} a_{jk}z_{jk},\ \ u_k \le \ell_k;\ \forall k\in K$; $\sum_{k\in K} z_{jk} \le x_j\ \forall j\in \goods$; $\mz\ge0$. That is, an LF function is given by 
a set of linear functions $u_k$,  with a maximum utility limit $\ell_k$ on each $k\in K$. LF captures interesting non-separable PLC functions, and includes SPLC  as a special case where each segment $k\in K$ has exactly one $j$ with $a_{jk} >0$.   
\item \emph{Generalized network utilities}: We introduce a new and broad class of utility functions defined by generalized flow instances. Here, the goods correspond to source nodes in a flow network with supply $x_j$, $j\in\goods$. Each arc has a capacity and a gain factor, and the utility $u(\mx)$ is the maximum flow value that can reach a designated sink node $t$ from supplies $\mx$. This value can be computed in polynomial time. The precise definitions are given in Section~\ref{sec:utility}. Additionally,
LF can be formulated as a special case of a network with two layers, the middle layer between the sources and the sink corresponding to the segments $K$. 
\end{itemize}

We note that generalized flows is a classical network flow model that can model commodity transportation with losses, as well as conversions between different types of entities such as in financial networks; see \citet[Chapter 15]{amo} for a description and various applications.

We show the following:
\begin{restatable}{theorem}{networkgale}\label{thm:network-gale}
Generalized network utility functions are Gale-substitutes.
\end{restatable}

\medskip

Theorem~\ref{thm::main-result} already shows that in Fisher markets with non-satiable generalized network utilities, Nash welfare maximizing allocations give a good approximation to utilities at competitive equilibrium.
However, satiable generalized network utilities may not satisfy the stronger $\Sigma$-Gale-substitutes property required in Theorem~\ref{thm::main-result}. Example~\ref{exp:leontief-free-satiated} shows that the theorem does not hold if only assuming the Gale substitutes property. For satiable utility functions, it is natural to look at \emph{thrifty (strong) equilibria}, a strengthening of CE, which are also guaranteed to exist, see, e.g.,~\cite{garg2022approximating}. In such equilibria, each agent receives an optimal bundle at the cheapest possible price (see Definition~\ref{def::thrifty-market-equilibrium}). 
An appealing property of this concept is that thrifty equilibria are Pareto efficient, whereas this may not hold for (weak) competitive equilibrium. 
We also show that the statement of Theorem~\ref{thm::main-result} remains true for generalized network utilities and for thrifty equilibria.
\begin{restatable}{theorem}{networkgalebound} \label{thm::satiate-gale}
    Consider a Fisher market instance where all utility functions are generalized network utilities.
    Then, for any thrifty competitive equilibrium $((\mx_i)_{i\in \agents}, \mathbf{p})$ and any Nash welfare maximizing allocation $(\my_i)_{i\in \agents}$, $u_i(\my_i) \geq \frac{1}{2} u_i(\mx_i)$ holds for every $i\in\agents$.
\end{restatable}

\begin{Remark}
    We note that no algorithm can guarantee agents at least $(1-\varepsilon)$ times their maximum CE utility for any $\varepsilon>0$. This is due to the fact that, as shown in Example~\ref{eg:1}, multiple thrifty equilibria may exist, each providing different utilities to the agents. The existence of an allocation that $\varepsilon$-approximates the best CE utility for each agent would, by compactness, imply the existence of an allocation where every agent achieves their best CE utility. However, this would contradict the fact that thrifty equilibria are Pareto efficient.
\end{Remark}

\subsection{Related work}
This paper strengthens the relationship between the two fundamental concepts of competitive equilibrium and maximum Nash welfare. Both are intensely studied concepts with a variety of applications, so we briefly mention further related results. As mentioned earlier, computing equilibria in Fisher markets is polynomial time for homogeneous and WGS utility functions, and beyond that, the problem is essentially PPAD-hard. For SPLC utilities, achieving a $\frac{1}{11}$-approximate-CE is PPAD-hard~\cite{DeligkasFHM24}. We note that~\cite{ChenT09,DeligkasFHM24,VaziraniY11} use a different definition of approximate-CE, where every agent gets an optimal bundle, but the market clears approximately. Our approximate-CE definition also appears in several previous works, e.g., in~\cite{codenotti2005market,codenotti2005polynomial,CodenottiVChapter07}. 

A variant of the Gale substitutes property also appeared in \cite{garg2023auction-arxiv}\footnote{The technical report version of \cite{garg2023auction}.}. This was used to obtain an auction algorithm for the maximum Nash welfare problem with \emph{capped} SPLC utilities, and the proof uses arguments similar to those in Appendix~\ref{sec:separable}. Similar auction algorithms could be developed more generally for Gale substitutes functions.

Fisher markets are a special case of more general Arrow-Debreu (exchange) markets. Computing equilibria in exchange markets is polynomial time for WGS utility functions (see, e.g.,~\cite{bei2019ascending,codenotti2005polynomial,garg2023auction}) and for CES utilities with parameter $-1 \le \rho < 0$~\cite{CodenottiMPV05}, and beyond that the problem is essentially PPAD-hard~\cite{ChenDDT09}, even for homogeneous functions, such as Leontief and CES with parameter $\rho < -1$~\cite{ChenPY13,CodenottiSVY06}. 

The maximum Nash welfare is a classical objective for allocating goods to agents. It has received significant attention in the literature on social choice and fair division; see, e.g.,~\cite{Moulin03}.  It provides remarkable fairness guarantees also for indivisible goods: under additive valuations, maximum Nash welfare allocations attain the ``unreasonable'' fairness guarantee of envy-freeness up to one good (EF1)~\cite{CaragiannisKMPS19}. Furthermore, for more general subadditive valuations, the fairness guarantee of $\tfrac12$-approximate envy-freeness up to any good (EFX) can be achieved at a $\tfrac23$-approximate maximum Nash welfare allocation~\cite{FeldmanMP23}.  

Another interesting line of research proposes Nash-bargaining-based solutions for allocation problems, offering efficient computation and other desirable structural properties; see, e.g.,~\cite{vazirani2012notion, HosseiniV22, GargTV23}. Recently, \citet{TrobstV2024} showed that computing a Pareto-efficient and envy-free allocation for matching utilities is PPAD-hard, and that the NSW solution gives a PO an  $2$-envy-free allocation. This latter result can be viewed as  a special case of our Theorem~\ref{thm:1}.
\medskip

\paragraph{Overview} 
The remainder of the paper is organized as follows. Section~\ref{sec:prelim} defines all models and definitions.  
Section~\ref{sec:poa} proves Theorem~\ref{thm:4}, establishing the tight bound of $(1/\ee)^{1/\ee}$ on the efficiency of competitive equilibria. In Section~\ref{sec:nash_welfare_approximates_equilibrium}, we show that Nash welfare maximizing allocations yield $2$-demand approximate-CE. Section~\ref{sec:agentwise_utility_guarantee} proves Theorem~\ref{thm::main-result}, showing the stronger guarantees of the maximum Nash welfare allocations for Gale-substitutes utility functions. Section~\ref{sec:gale_substitute_utilities} proves Theorem~\ref{thm:sep-gale}, showing that separable utilities are $\Sigma$-Gale substitutes, and considers generalized network utility functions. 

Appendix~\ref{sec:galedemand} presents several properties of the Gale demand correspondence. Appendix~\ref{sec:main-proof} provides the complete proof of Theorem~\ref{thm::main-result}. Appendix~\ref{sec:separable} proves Theorem~\ref{thm:sep-gale} for the general case, and Appendix~\ref{sec:network} contains the proofs of Theorems~\ref{thm:network-gale} and \ref{thm::satiate-gale}. Additional examples are presented in Appendix~\ref{sec:egs}.


\section{Models and definitions}\label{sec:prelim}
For a function $f\,:\,X\to Y$, let $\dom(f)\subseteq X$ denote its domain. For $\mx\in\R^n$, let $\supp(\mx)=\{i\, :\, x_i\neq 0\}$ denote its support. For $\mx,\my\in\R^n$, we use $\pr{\mx}{\my}=\mx^\top \my$ for the standard scalar product.

\paragraph{Subgradients and supergradients} The functions we consider in this paper are not always differentiable. We will use sub- and supergradients, defined as follows.
\begin{Definition}
    Given a function $f\, :\, \R^n\to\R$, a vector $\mg\in\R^n$ is a \emph{supergradient} of $f$ at the point $\mx\in\dom(f)$ if for all $\my\in\dom(f)$, $f(\mx) + \pr{\mg}{\my-\mx} \geq f(\my)$. Similarly, $\mg\in\R^n$ is a \emph{subgradient} of $f$ at $\mx\in\dom(f)$ if $f(\mx) + \pr{\mg}{\my - \mx}\le f(\my)$ for all $\my\in\dom(f)$. We denote the set of all supergradients of $f$ at $\mx$ as $\supg{f} (\mx)$,  set of all subgradients as $\subg{f} (\mx)$.
\end{Definition}
Thus, $\mg$ if a supergradient of $f$ if and only if $-\mg$ is a subgradient of $-f$. We recall that a function $f$ is convex if and only if $\subg{f}(\mx)\neq\emptyset$ for every $\mx\in\dom(f)$ and concave if and only if $\supg{f}(\mx)\neq\emptyset$ for every $\mx\in\dom(f)$. For a differentiable convex (concave) function $f$, the unique subgradient (supergradient) at every point is the gradient $\nabla f(\mx)$. 
\subsection{Utility functions}\label{sec:utility}
Given a set $\goods$  of divisible goods, we will focus on utility functions with the following properties.

\begin{Definition}\label{def:utility-function}
A function $u(\cdot)\, , \, \Rp^\goods\to \R$ is a \emph{utility function} if $u(\0)=0$, $u(\cdot)$ is concave, continuous, monotone non-decreasing, and there exists an $\mx\in\Rp^\goods$ such that $u(\mx)>0$.
\end{Definition}
We say that a utility function is \emph{separable} if $u(\mx)=\sum_{j\in\goods} v_j(x_j)$. 
We will also distinguish between satiable and non-satiable utilities.
\begin{Definition}
    A utility function $u(\cdot)$ is \emph{non-satiable} if and only if, for any $\mx\in\Rp^\goods$, there exists  $\mx'\in\Rp^\goods$, such that $u(\mx') > u(\mx)$. Otherwise, we say the utility function is \emph{satiable}.
\end{Definition}

\paragraph{Demand correspondence and Gale demand correspondence}
 By a  \emph{price vector}, we mean a vector $\mp\in\Rp^\goods$. Given prices $\mp$ and a budget $b\in\Rp$, we define the 
 \emph{demand correspondence} as the set of maximum utility bundles achievable at $\mp$ and $b$:
  \begin{equation}\label{def:demand}\demand{}^u(\mathbf{p},b) \defeq \arg\max_{\mx\in\Rp^\goods}\left\{ u(\mx)\, :\, \pr{\mp}{\mx}\le b\right\}\, .
  \end{equation}
  Note that if $u$ is a non-satiable utility function, then $\pr{\mp}{\mx}=b$ must hold for every $\mx\in\demand{}^u(\mp,b)$.
  
$\demand{}^u(\mathbf{p},b)$ is always a non-empty compact set when $\mp > 0$. However, in cases when some prices may be  $0$, the supremum of utility over the budget constraint might be $\infty$, resulting in unbounded demand. This may occur even when the supremum is finite, causing technical difficulties.\footnote{For example, given $u(x_1, x_2) = 1 - \frac{1}{x_1 + x_2 + 1}$, $\mp = (1, 0)$, and $b = 1$, $\sup_{\mx\in\Rp^\goods}\left\{ u(\mx)\, :\, \pr{\mp}{\mx}\le b\right\} = 1$. This value is approached as $x_1 = 1$ and $x_2 \rightarrow \infty$, but it is never achieved by any finite bundle $\mx$.} To circumvent these issues, we adopt the following assumption regarding utility functions $u(\cdot)$:
  \begin{Assumption}\label{assum:basic-demand-system}
      For any $\mp \geq \mathbf{0}$ and $b > 0$, if there is a sequence of allocations $\{\mx^{(k)}\}_{k=1}^\infty$ such that $\pr{\mp}{\mx^{(k)}}$ converges to $b$ and $u(\mx^{(k)})$ converges to some finite value $M < \infty$, then there exists a finite allocation $\mx^* \in\Rp^\goods$ such that $u(\mx^*) = M$ and $\pr{\mp}{\mx^*}\le b$. 
  \end{Assumption} 

This assumption always holds if the set of possible allocations is restricted to a compact domain $Z\subseteq \Rp^\goods$; note that \citet{nesterov2018computation} make  such a compactness assumption. For the unrestricted domain $\Rp^\goods$, Assumption~\ref{assum:basic-demand-system} can be shown to hold for arbitrary piecewise linear utility functions, in particular, for generalized network utilities defined below.

The \emph{Gale demand correspondence} is defined as
 \begin{equation}\label{def:galedemand}
 \galedemand{}^u(\mathbf{p},b) \defeq \arg\max_{\my\in\Rp^\goods} b\log u(\my)- \pr{\mp}{\my}\, .
 \end{equation}
We define the domain of $\galedemand{}^u$ with respect to prices as \footnote{An equivalent formulation of the definition is given by:
\begin{equation}\notag
\FQ{u} \defeq \left\{ \mp\in\Rp^\goods ~:~ \forall b > 0, \sup_{\my\in\Rp^\goods}  b \log u(\my) - \pr{\mp}{\my} < +\infty\right\} .
\end{equation}
This definition is obtained by substituting the universal quantifier ($\forall$) for the existential quantifier ($\exists$). Lemma~\ref{lem::gale-domain-property} \eqref{lem::arb-budget} provides further details on this equivalence.}
\begin{equation}\label{def:FQ}
\FQ{u} \defeq \left\{ \mp\in\Rp^\goods ~:~ \exists b > 0, \sup_{\my\in\Rp^\goods}  b \log u(\my) - \pr{\mp}{\my} < +\infty\right\}\, .
\end{equation} 
Under Assumption~\ref{assum:basic-demand-system}, 
$\0\in \FQ{u}$ if and only if $u(.)$ is satiable. Every strictly positive $\mp>\0$ is in $\FQ{u}$, and 
if $u$ is satiable, then $\FQ{u} = \Rp^\goods$ (see Lemma~\ref{lem::gale-domain-property}). 
In Appendix~\ref{sec:galedemand}, we show the following.
\begin{restatable}{lemma}{galedbasic}\label{lem-galedbasic}
   Under Assumption~\ref{assum:basic-demand-system}, for any $\mp \in \FQ{u}$ and for any $b > 0$, it holds that $\galedemand{}^u(\mathbf{p},b)\neq\emptyset$.
\end{restatable}
Gale demands may behave very differently from the standard demand correspondence. Example~\ref{eg:mono} shows the counter-intuitive property increasing the prices $\mp$ may actually increase the Gale utility value  $\max_{\my\in\Rp^\goods} b\log u(\my)- \pr{\mp}{\my}$.

Recall the notions of bounded linear, separable piecewise linear (SPLC) and Leontief-free (LF) utilities from the Introduction. We now define the class of generalized network utility functions.

\paragraph{Generalized network utilities}
We now formally introduce the generalized flow model. 
A \emph{generalized flow instance} is given by a directed graph $G=(V,E)$  with a sink node $t\in V$, a supply set $S\subseteq V\setminus \{t\}$, a capacity vector $\mcp\in \Rp^E$, and positive gain factors $\bm{\gamma}\in\Rpp^E$. 
We let $\din(S)$ and $\dout(S)$ denote the set of arcs entering and leaving a subset  $S\subseteq V$, respectively. For a node $v\in V$, we use $\din(v)\defeq\din(\{v\})$ and $\dout(v)\defeq\dout(\{v\})$. 
For a vector $\mf\in \Rp^E$, for every arc $e=(v,w)$, we interpret $f_e$ as the flow entering $e$ at node $v$.  The flow gets multiplied by the gain factor $\gamma_e$ while traversing the arc, with $\gamma_e f_e$ amount reaching node $w$. Thus,
the \emph{net flow} at node $v\in V$ is defined as 
\[
\net{\mf}{v}\defeq \sum_{e\in \din(v)} \gamma_e f_e -\sum_{e\in\dout(v)} f_e\, .
\]
Given a \emph{supply vector} $\mx\in \Rp^S$, we say that $\mf\in \Rp^E$ is a \emph{feasible generalized flow with supply $\mx$} and capacity $\mcp$ if $\0\le \mf\le \mcp$, and $\net{\mf}{v}\ge-x_v$ for all $v\in S$, and $\net{\mf}{v}\ge0$ for all $v\in V\setminus (S\cup \{t\})$. The \emph{maximum generalized flow value from supply $\mx$} is the maximum amount of $\net{\mf}{t}$ among all feasible flows with supply $\mx$. It corresponds to the optimum value of the following linear program:
\begin{equation}\label{eq:genflow-LP}\tag{$\Genflow(\mx)$}
\begin{aligned}
\max~& \net{\mf}{t}\\
\net{\mf}{v}&\ge -x_v\quad\forall v\in S\\
\net{\mf}{v}&\ge 0\quad\forall v\in  V\setminus (S\cup \{t\})\\
\0&\le \mf\le \mcp\, .
\end{aligned}
\end{equation}

\begin{Definition} The function $u: \Rp^\goods\to \Rp$ is a \emph{generalized network utility} function if 
$u(\mx)$ is the maximum generalized flow value from supply $\mx$ in some generalized flow instance where the supply set is $S=\goods$,  the maximum flow amount is finite for any supply vector, $u(\0)=0$, and there exists $\mx\in\Rp^\goods$ with $u(\mx)>0$.
\end{Definition}

The following properties are immediate from the definition and show that these are indeed utility functions in accordance with Definition~\ref{def:utility-function}.
\begin{proposition}\label{prop:gnuf-basic}
If $u: \Rp^\goods\to \Rp$ is a generalized network utility function, then $u(\mx)$ is monotone non-decreasing and concave, and satisfies Assumption~\ref{assum:basic-demand-system}.
\end{proposition}

\paragraph{SPLC and Leontief-free as special cases}
An SPLC function can be represented as a generalized network utility function by setting $V=\goods\cup\{t\}$, and for each $j\in\goods$, adding a capacitated arc $e=(j,t)$ for each linear segment of $v_j$, where $\gamma_e$ is the marginal utility on this segment and $h_e$ is the length of the segment. Similarly, we can represent a Leontief-free utility as a generalized network utility function by setting $V=\goods\cup K\cup \{t\}$; recall $K$ is the set of segments. For each $j\in\goods$ and $k\in K$, we add an arc $e=(j,k)$ with gain factor $\gamma_e=a_{jk}$ and capacity $h_e=\infty$. For each segment $k\in K$, we add an arc $e=(k,t)$ with $\gamma_e=1$ and $h_e=\ell_k$. 

\paragraph{
Utilities from production} Assume that agents are competing for access to $m$ machines $\goods$ that can produce $\ell$ different products. The machine $j$ can produce $a_{jk}$ units of product $k$ per time unit. The unit revenue of product $k$ is $r_k$. There are also limits $t_{jk}$ on the quantity of each product $k$ that the agent can produce on machine $j$, as well a limit $d_k$ on the maximum possible amount producible from $k$. Given $x_j$ time on machine $j\in\goods$, the utility of the agent is the maximum possible revenue subject to these constraints. This can be formulated using a two-layer generalized flow network, similar to Leontief-free case described above, but with additional capacity bounds $t_{jk}$. See Figure~\ref{fig:utilities-production} for an illustration. The model allows for more complex production networks with possibly multiple layers corresponding to raw materials, intermediary and final products.

\begin{figure}[t]
\[\begin{tikzcd}[column sep=5ex,row sep=1ex]
	{x_1 \space\bullet} &&& {1\space \bullet} \\
	{x_2 \space \bullet} & {} && \vdots \\
	\vdots &&& {k \space \bullet} &&& { \bullet} \\
	{x_i \space \bullet} &&& \vdots \\
	\vdots &&& {\ell \space \bullet} \\
	{x_m \bullet}
	\arrow[from=1-4, to=3-7]
        \arrow[from=1-1, to=3-4]
	\arrow[from=2-1, to=3-4]
	\arrow["{(r_k, d_k)}", from=3-4, to=3-7]
	\arrow[from=2-1, to=1-4]
	\arrow["{\quad \quad \quad (a_{jk}, t_{jk})}", from=4-1, to=3-4]
	\arrow[from=4-1, to=5-4]
	\arrow["\vdots"', from=5-4, to=3-7]
\end{tikzcd}\]
\caption{Utilities from production.}\label{fig:utilities-production}
\end{figure}

\subsection{Competitive equilibrium} 
In a \emph{Fisher market instance}, we are given a set of agents $\agents$  and a set of divisible goods $\goods$, with one unit available of each good. Agents have utility functions $u_i\,:\,\Rp^\goods\to\Rp$ and budgets $b_i\in\Rpp$, $i\in\agents$; we denote the budget vector as $\mb = (b_i)_{i\in\agents}$. An allocation is represented by $(\mx_i)_{i\in \agents}$, where $x_{ij}$ denotes the amount of good $j$ allocated to agent $i$. A  \emph{price vector} is a vector $\mp\in\Rp^\goods$. The demand correspondence of agent $i$ will be denoted as
 $\demand{i}(\mathbf{p},b_i)$. 
 
\begin{Definition}[Competitive  equilibrium] Given a Fisher market instance as above, the allocations and prices $((\mx_i)_{i\in \agents}, \mp)$ form a \emph{competitive (market) equilibrium} if the following hold:
\begin{enumerate}[(i)]
\setlength\itemsep{0em}
    \item every agent gets an optimal utility at these prices: $\mx_i\in \demand{i}(\mathbf{p},b_i)$ for every agent $i\in\agents$. 
    \item no good is oversold: $\sum_i x_{ij} \leq 1$ for all $j\in\goods$.
    \item every good with positive price is fully sold:  $\sum_i x_{ij} = 1$ if $p_j>0$.
\end{enumerate}
\end{Definition}
In a \emph{demand-approximate competitive equilibrium}, agents may  have suboptimal bundles:
\begin{Definition}[Approximate competitive equilibrium]\label{def:approx-market-eq}
Given a Fisher market instance and $\approxi\ge 1$, the allocations and prices $((\mx_i)_{i\in \agents}, \mp)$ form an \emph{$\approxi$-demand approximate competitive equilibrium}  
if
\begin{enumerate}[(i)]
\setlength\itemsep{0em}
    \item $\pr{\mp}{\mx_i}\le b_i$ for each $i\in\agents$, and $i$ gets at least $\approxi$-fraction of the optimal utility at these prices: $u_i(\mx_i)\ge \frac{1}\approxi\max_{\mx_i\in\Rp^\goods}\{ u_i(\mx_i)\, :\, \pr{\mp}{\mx_i}\le b_i \}$.
     \item  no good is oversold: $\sum_i x_{ij} \leq 1$ for all $j\in\goods$.
    \item every good with positive price is fully sold:  $\sum_i x_{ij} = 1$ if $p_j>0$.
\end{enumerate}
\end{Definition}

It is well known that when all $b_i$'s are identical, competitive equilibria yield an allocation that is exactly envy-free. Here, we demonstrate that our definition of approximate competitive equilibrium also results in approximate envy-freeness.
\begin{Definition}
    For a given allocation $(\mx_i)_{i\in \agents}$, we say it is $\beta$-envy-free if, for any agents $i$ and $j$, the condition $u_i(\mx_i) \geq \frac{1}{\beta} u_i(\mx_j)$ holds.
\end{Definition}
\begin{lemma} \label{lem::approx-envy-free}
   When $b_i$'s are identical, the allocation $(\mx_i)_{i\in \agents}$ of an $\approxi$-approximate competitive equilibrium is also $\approxi$-envy-free. 
\end{lemma}
\begin{proof}
    This follows directly from $u_i(\mx_i)\ge \frac{1}\approxi\max_{\mx_i\in\Rp^\goods}\left\{ u_i(\mx_i)\, :\, \pr{\mp}{\mx_i}\le b_i \right\} \ge \frac{1}{\approxi} u_i(\mx_j)$.
\end{proof}

\paragraph{Thrifty equilibrium}
In the case of satiable utility functions, it is possible for an (exact) equilibrium to satisfy $\pr{\mp}{\mx_i}<b_i$. Hylland and Zeckhauser~(\citeyear{hylland1979efficient}) showed that competitive equilibria may not be Pareto efficient in the presence of satiable utilities. Mas-Colell~(\citeyear{mas1992equilibrium}) strengthened the definition of competitive equilibrium by requiring that each agent select a least-cost optimal bundle. This is also known as a \emph{strong equilibrium}.

\begin{Definition}[Thrifty competitive equilibrium]\label{def::thrifty-market-equilibrium}
Given a Fisher market instance, the allocations and prices $((\mx_i)_{i\in \agents}, \mp)$ form a \emph{thrifty competitive equilibrium}  if $((\mx_i)_{i\in \agents},\mp)$ is a competitive equilibrium, and moreover,  for every $i\in\agents$, the allocation $\mx_i$ minimizes $\pr{\mp}{\mx_i}$ among  all allocations $\demand{i}(\mathbf{p},b_i)$.
\end{Definition}

\begin{lemma}[\citet{garg2022approximating,gul2022the}]
    Thrifty competitive equilibria always exist and are Pareto efficient.
\end{lemma}

Note that when all utility functions are non-satiable, every competitive equilibrium is thrifty.

\subsection{Nash welfare}
\begin{Definition}
Given a Fisher market instance, the \emph{Nash welfare} of an allocation $(\my_i)_{i\in \agents}$ is the weighted geometric mean of their utilities, namely, 
\[ \left(\prod_{i\in\agents} u_i(\my_i)^{b_i}\right)^{\frac{1}{\sum_{i\in \agents} b_i}}\, .  \]
The allocation $(\my_i)_{i\in \agents}$ \emph{maximizes the (weighted) Nash welfare}, if it achieves the highest Nash welfare among all allocations satisfying $\sum_{i\in\agents} y_{ij} \leq 1$ for all $j\in\goods$. 
\end{Definition}
It is well known that the Nash welfare maximizing allocations correspond exactly to the optimal solutions of the convex program introduced by~\citet{eisenberg1959consensus}. 
\begin{equation}\label{def:log-NSW}\tag{EG} 
\begin{aligned}
   \max &\sum_{i\in\agents} b_i \log u_i(\my_i) \ \ 
    \text{s.t.} &\sum_{i\in\agents} y_{ij} \leq 1, \ \forall j\in\goods \ \ \text{and} 
    &~~~y_{ij}\ge0, \ \forall i\in\agents, j\in\goods\, .
\end{aligned}
\end{equation}
The Lagrangian dual of \eqref{def:log-NSW} can be written as 
\begin{equation}\label{eq:Gale-dual-function}
    \EGdual (\mq) \defeq \sum_{i\in\agents} \max_{\my_i\in\Rp^{\goods}}\left( b_i\log u_i(\my_i) - \pr{\mq}{\my_i}\right) + \sum_{j\in \goods} q_j.
\end{equation}
We define the domain $\dom(\EGdual)\defeq \cap_{i\in\agents} \FQ{u_i}$ (see \eqref{def:FQ}). Thus, $\dom(\EGdual)$ is the set of nonnegative prices $\mq$ such that $\EGdual(\mq)<\infty$.
Slater's condition is satisfied; therefore strong duality holds, and the dual optimum is attained. Consequently, the optimal Nash welfare value is equal to 
\begin{equation}\label{eq:Gale-dual}
\min_{\mq\in\dom(\EGdual)} \EGdual(\mq)\, ,
\end{equation}
and 
we refer to the optimal $\mq\in\dom(\EGdual)$ as \emph{Gale prices}. We note that these may not be unique. 
To justify this term, let us 
recall the definition of the Gale demand correspondence \eqref{def:galedemand}. The term
corresponding to agent $i$ in the definition of $\EGdual (\mq)$ is the value of the assignments in the Gale demand correspondence of this agent at prices $\mq$ and budget $b_i$. We denote by $\galedemand{i}(\mq,b_i)$ the Gale demand corresponding to the utility function $u_i$, $i\in\agents$.

\subsection{Gale substitute utilities}
We now define \emph{Gale substitute} and \emph{$\Sigma$-Gale substitute} utilities, the crucial classes where we show that 
 the Nash welfare maximizing allocation recovers at least half of the utility of every agent in any possible competitive equilibrium (see Theorem~\ref{thm::main-result}).

\begin{Definition}\label{def:gale-subs}
    The utility function $u(\cdot)$ is \emph{Gale-substitutes} if it satisfies the following properties:
\begin{enumerate}[(i)]
\setlength\itemsep{0em}
        \item for any  $b > 0$, any $\mathbf{q},\mq' \in \FQ{u}$, and any $\my \in \galedemand{}^u (\mq,b)$, if  $\mathbf{q} \geq \mathbf{q}'$, then there exists a $\my' \in \galedemand{}^u (\mathbf{q}',b)$ such that $y'_{j} \leq y_{j}$ whenever $q_j' = q_j$; \label{def:gale-subs-1}
        \item for any  $b \geq b' > 0$ and $\mathbf{q} \in \FQ{u}$, and for any $\my \in \galedemand{}^u (\mathbf{q},b)$, there exists a $\my' \in \galedemand{}^u (\mq,b')$ such that $\my'\leq \my$. \label{def:gale-subs-2}
    \end{enumerate}
    We say that $u(\cdot)$ is $\Sigma$-Gale-substitutes, if it satisfies the Gale-substitutes condition along with the following additional property:
    \begin{equation}\label{prop:satiate}\tag{$\Sigma$}
    \mbox{if }\mathbf{0} \in \FQ{u}\,,\, b>0\mbox{ and }\mathbf{q} \in\Rp^\goods,\mbox{ then for any }\my \in \galedemand{} (\mathbf{0},b), \exists \my' \in \galedemand{}^u (\mathbf{q},b)\, :\, \my' \leq \my\, .
\end{equation}
\end{Definition}

\begin{Remark}
As noted above, $\mathbf{0} \in \FQ{u}$ if and only if $u$ is satiable; thus, \eqref{prop:satiate} is vacuously true for non-satiable utilities. In the satiable case, it asserts that if $\my$ is a bundle of maximum possible utility, then at any prices there is an optimal bundle below $\my$. Observe that this is immediate for separable utilities. On the other hand,  \eqref{prop:satiate} fails for the simple Leontief-free utility $u(\mx)=\max\{x_1+x_2,1\}$: $\my=(1,0)$ is an optimal bundle at $\mq=\0$, whereas for $b=1$, $\mq=(2,1)$, the unique solution in $\galedemand{}^u (\mathbf{q},b)$ is $\my'=(0,1)$.
\end{Remark}

Next, we present a utility function that satisfies the Gale-substitutes property but does not exhibit the weak gross substitutes (WGS) property.
\begin{example}
Consider an agent with budget $b$ and a utility function over two goods given by $u(\mx) = \min\{x_1, 0.4\} + x_2$. This utility is separable and, as such, Gale-substitutes (see  Theorem~\ref{thm:sep-gale}). However, we show that this is not WGS: there exist prices $\mp\le \mp'$ and $\mx\in\demand{}^u(\mathbf{p},b)$ such that for some good $j$ with $p_j=p'_j$, every $\mx'\in \demand{}^u(\mathbf{p},b)$ has $x'_j<x_j$.
That is, let $\mp=(0.8,1)$ and $\mp'=(0.9,1)$, and $b=1$. It is easy to see that $\mx=(0.4,0.68)$ is the unique optimal bundle at prices $\mp$, and $\mx'=(0.4,0.64)$ is the unique optimal bundle at prices $\mp'$ in contradiction to the WGS property.

We note that the gross substitutes property is one of the key conditions that allows for the efficient computation of market equilibria. However, SPLC utilities do not satisfy the gross substitutes property, as illustrated in the above example. In a later section, we show that SPLC preferences are Gale-substitutes, which forms the foundation of our main results.
\end{example}

\section{The price of anarchy bound}\label{sec:poa}
We now show the `price of anarchy' result stated in Theorem~\ref{thm:4}. The proof employs an argument based on concavity and the log-sum inequality, drawing inspiration from Theorem 1 in \cite{BranzeiGM22}.

\begin{theorem}[Restatement of Theorem~\ref{thm:4}] \label{thm::nsw-bound-2}
Given a Fisher market instance with concave utility functions, consider any competitive equilibrium $((\mx_i)_{i\in \agents}, \mathbf{p})$ and a Nash welfare maximizing allocation $(\my_i)_{i\in \agents}$. Then,
    \begin{align*}
        \left(\prod_{i\in\agents} u_i(\mx_i)^{b_i}\right)^{\frac{1}{\sum_{i\in \agents} b_i}} \geq \left(\frac{1}{\ee}\right)^{\frac{1}{\ee}} \left(\prod_{i\in\agents} u_i(\my_i)^{b_i}\right)^{\frac{1}{\sum_{i\in \agents} b_i}}.
    \end{align*}
\end{theorem}

\begin{proof}
Without loss of generality, we can assume $\sum_{i\in \agents} b_i=1$.
    Note that since $((\mx_i)_{i\in \agents}, \mathbf{p})$ is a competitive equilibrium, then $u_i(\mx_i) \geq \min \left\{1, \frac{b_i}{\pr{\mp}{\my_i}} \right\}u_i(\my_i)$, as agent $i$ can use their budget, $b_i$, to buy at least $\frac{b_i}{\pr{\mathbf{p}}{\my_i}}$ fraction of $\my_i$. The lower bound holds because if $b_i<\pr{\mp}{\my_i}$, then $u_i\left(\frac{b_i}{\pr{\mp}{\my_i}} \my_i\right)\ge \frac{b_i}{\pr{\mp}{\my_i}}u_i(\my_i)$  using concavity and $u_i(\0)=0$. 
Thus,
    \begin{align*}
        \sum_{i\in\agents} {b_i} \log \frac{u_i(\my_i)}{u_i(\mx_i)} \leq  \sum_{i\in\agents} {b_i} \log \max \left\{1, \frac{\pr{\mp}{\my_i}}{b_i} \right\}.
    \end{align*}
    Let $b'_i \defeq {\pr{\mp}{\my_i}}$, and let $\agents'$ be the set of agents such that $b'_i \geq b_i$, $B \defeq \sum_{i \in \mathcal{A}'} b_i$, and $B' \defeq \sum_{i \in \mathcal{A}'} b'_i$. Therefore, using the log sum inequality,
    \begin{align*}
           \sum_{i\in\agents} {b_i} \log \frac{u_i(\my_i)}{u_i(\mx_i)} \leq \sum_{i \in \agents'} {b_i} \log \frac{b'_i}{b_i} \leq B \log \frac{B'}{B} \, .
    \end{align*}
   By definition, $B\le \sum_{i\in\agents} b_i=1$. We also note that $B'\le \sum_{i\in\agents} b_i'=\sum_{i\in \agents} \pr{\mp}{\my_i}\le \sum_{j\in \goods} p_j \le \sum_{i\in\agents} b_i \leq 1$. 
    \begin{align*}
         \sum_{i\in\agents} {b_i} \log \frac{u_i(\my_i)}{u_i(\mx_i)} \leq \max_{B \leq 1} B \log \frac{1}{B} \leq \frac{1}{\ee}\, .
    \end{align*}
    The statement follows.
\end{proof}

The following example demonstrates that the bound $(1/\ee)^{1/\ee}$ is tight.

\begin{example}[Tight example on price of anarchy]\label{ex:poa}
Consider an instance with two agents and two goods. The first agent has a budget of $1$ and has a utility function: 
\[ u_1(\mx_1) = x_{11} + \varepsilon x_{12}, \] 
and the second agent has a budget of $\ee - 1$ and has a utility function 
\[ u_2(\mx_2) = \varepsilon x_{21} + \min\{x_{22}, 1 - \varepsilon\}, \]
where $\varepsilon > 0$ is sufficiently small. 
The maximum Nash welfare is at least $(1 - \varepsilon)^{(\ee-1)/\ee}$, taken at 
$\my_1 = (1, 0)$ and $\my_2=(0, 1)$. 
Conversely, the unique competitive equilibrium is characterized by prices $p_1 = \ee (\frac{1}{1 + \varepsilon})$ and $p_2 = \ee (\frac{\varepsilon}{1 + \varepsilon})$, with corresponding allocations given by:
\[ \mx_1 = \left( \frac{1 + \varepsilon - \ee \varepsilon^2}{\ee}, \varepsilon \right), \quad \mx_2 = \left( 1 -  \frac{1 + \varepsilon - \ee \varepsilon^2}{\ee}, 1 - \varepsilon \right). \] 
The Nash welfare objective for this allocation is $\left(  \frac{1 + \varepsilon}{\ee} \cdot \left[ 1 - \varepsilon + \varepsilon \cdot \left( 1 -  \frac{1 + \varepsilon - \ee \varepsilon^2}{\ee}\right) \right]^{\ee - 1} \right)^{\frac{1}{\ee}}$.
When $\varepsilon \to 0$, the Nash welfare ratio between these two allocations converges to $(1/\ee)^{1/\ee}$.
\end{example}

The above example also shows that the \emph{price of stability} is the same as the price of anarchy, as there is a unique CE.

\section{Nash welfare approximates competitive equilibrium}\label{sec:nash_welfare_approximates_equilibrium}
We now derive Theorem~\ref{thm:1}, showing that Nash welfare maximizing allocations can be seen as approximate competitive equilibria. Our first simple lemma shows that the Gale demand solutions are feasible with respect to the budget:
\begin{lemma}\label{lem::budget-spending-NSW}
{\color{blue}}Given any  utility function $u\,:\, \Rp^\goods\to \Rp$, $b>0$ and $\mq\in \Rp^\goods$,
for  any $\my\in\Rp^\goods$ such that $\pr{\mathbf{q}}{\my} > b$, there exists $\my'\in\Rp^\goods$ such that $\pr{\mathbf{q}}{\my'}\leq b$ and  $b \log u(\my') - \pr{\mathbf{q}}{\my'} >  b \log u(\my) - \pr{\mathbf{q}}{\my}$.
In particular, if  $\my\in\galedemand{}^u(\mathbf{q},b)$, then, $\pr{\mq}{\my}\le b$.
\end{lemma}
\begin{proof}
Let $\my' \defeq \my \cdot \frac{b}{\pr{\mq}{\my}}$. By concavity and $u(\0)=0$, we get $u(\my')\geq\frac{b}{\pr{\mq}{\my}} u(\my)$. Thus, 
\begin{align*}
    b\log u(\my') - \pr{\mq}{\my'} &\geq b \log u(\my) - b \log  \frac{\pr{\mq}{\my}}{b} - b 
    >  b \log u(\my) - \pr{\mq}{\my}\, ,
\end{align*}
where the second inequality used $\log(\alpha)+1 < \alpha$ for $\alpha >  1$.
\end{proof}
The next lemma follows directly from the subgradient Karush--Kuhn--Tucker (KKT) conditions used to characterize solutions in the Gale demand correspondence. For a detailed treatment of Lagrangian duality and KKT conditions for nondifferentiable functions, we refer the reader to Chapter 3 of \cite{ruszczynski2011}. 

\begin{lemma}\label{lem::gale-kkt}
For a  utility function $u\,:\, \Rp^\goods\to \Rp$, $b>0$ and $\mq\in \FQ{u}$, let $\my\in\galedemand{}^u(\mathbf{q},b)$. Then, there exists a supergradient $\mg\in \supg u(\my)$ such that
\[
g_j \le {q_j}\cdot \frac{u(\my)}{b}\, , \quad\forall j\in \goods\,,\,\mbox{with equality  if }y_{j} > 0\, ,
\] Conversely, if such a $\mg\in\supg u(\my)$ exists then $\my\in\galedemand{}^u(\mathbf{q},b)$.
\end{lemma}

\thmfirst*
\begin{proof}
    By Lemma~\ref{lem::budget-spending-NSW}, all agents satisfy the budget constraints, $\pr{\mq}{\my_i} \leq b_i$ for all $i$. Additionally, the market clears by the KKT conditions for \eqref{def:log-NSW}. For the approximate competitive equilibrium, we only need to show  $u_i(\my_i) \geq \frac{1}{2} \max_{\pr{\mq}{\mx_i} \leq b_i} u_i(\mx_i) $ for all agents.

    Note that $(\my_i)_{i \in \agents}$ is the allocation that maximizes Nash welfare, where each $\my_i$ maximizes $b_i \log u_i(\my_i) -\pr{\mq}{\my_i}$. By Lemma~\ref{lem::gale-kkt}, there is a supergradient $\mg_i$ such that $g_{ij} \leq q_j u_i(\my_i) / b_i$ with equality whenever $y_{ij} > 0$. Then, for any $\mx_i\in\Rp^\goods$ such that $\pr{\mq}{\mx_i} \leq b_i$, 
    \begin{align*}u_i(\mx_i) &\leq u_i(\my_i) + \pr{\mg_i}{\mx_i - \my_i}  &\text{(convexity of $u_i$)} \\
    &\leq u_i(\my_i) + \frac{u_i(\my_i)}{b_i} \pr{\mq}{\mx_i - \my_i}  &\text{($g_{ij} \leq q_j u_i(\my_i) / b_i$ with equality whenever $y_{ij} > 0$)} \\
    &\leq \left(2 - \frac{\pr{\mq}{\my_i}}{b_i}\right) u_i(\my_i)  &\text{($\pr{\mq}{\mx_i} \leq b_i$)}\\
    &\leq 2 u_i(\my_i)\, . &
    \end{align*}
    This proves that $((\my_i)_{i\in \agents}, \mq)$ forms a $2$-approximate competitive equilibrium.
    Additionally, the $2$-envy-freeness of the allocation follows directly from Lemma~\ref{lem::approx-envy-free}.  
\end{proof}

One may wonder if the $2$-demand-approximate CE can be improved. The following example demonstrates that this bound of $2$ is tight.
\begin{example}[Tight example on demand-approximate CE]
Consider the case with two agents and a single good. Both agents have unit budgets. The first agent's utility for this good is linear:
\[
u_1(x_1) = x_1.
\]
The second agent's utility is defined as:
\[
u_2(x_2) = 
\begin{cases}
\alpha x_2 & \text{for } x_2 \leq \frac{1}{\alpha}, \\
1 + \beta \left(x_2 - \frac{1}{\alpha}\right) & \text{for } x_2 > \frac{1}{\alpha},
\end{cases}
\]
where $\alpha$ is a large number and $\beta \leq \alpha$.

When $\beta = \frac{1 - \frac{1}{\alpha}}{1 - \frac{2}{\alpha}}$, with the optimal dual variable (price) $q$ equal to $q = \frac{1}{x_1} = \frac{1}{1 - \frac{1}{\alpha}}$, the Nash welfare maximizing allocation is 
\[
x_1 = 1 - \frac{1}{\alpha}, \quad x_2 = \frac{1}{\alpha}.
\]

Under this price, agent 2 can afford $x_2' = 1 - \frac{1}{\alpha}$ with the budget, resulting in a utility of
\[
u_2(x_2') = 1 + \beta \left( x_2' - \frac{1}{\alpha} \right) = 2 - \frac{1}{\alpha}.
\]

In contrast, in the Nash welfare maximizing allocation, agent 2 only receives a utility of $1$.
\end{example}

Although Theorem~\ref{thm:4} guarantees that the Nash welfare of any competitive equilibrium is at least a constant fraction of the optimum, the following example illustrates that agents can receive significantly different utilities across different equilibria.
\begin{example}[Utility discrepancy across competitive equilibria]\label{eg:2}
    Consider a market with $n$ agents and $n$ goods, where each agent has a unit budget. For $1 \leq i \leq n - 1$, agent $i$'s utility is given by
    \[
        u_i(\mx_i) = \min\{x_{ii}, 1\} + \varepsilon x_{in},
    \]
    for a small $\varepsilon > 0$, and agent $n$'s utility is
    \[
        u_n(\mx_n) = x_{nn}.
    \]
    There exist at least two distinct competitive equilibria:
    
    \begin{itemize}
        \item In the first equilibrium, all goods are priced at $1$, and each good $i$ is allocated entirely to agent $i$ for $1 \leq i \leq n$. Agent $n$ thus attains a utility of $1$. This equilibrium is also a Nash welfare maximizing allocation.
        
        \item In the second equilibrium, the first $n - 1$ goods are priced at $0$, while the $n$-th good is priced at $n$. Goods $1$ to $n - 1$ are allocated to agents $1$ through $n - 1$ respectively, and the $n$-th good is shared equally among all $n$ agents. As a result, agent $n$ obtains utility equal to $\frac{1}{n}$. 
    \end{itemize}
\end{example}

\section{Agent-wise utility guarantee}\label{sec:agentwise_utility_guarantee}
We now turn to the proof of the main result of this paper, Theorem~\ref{thm::main-result}, asserting that 
 for Gale-substitutes utility functions, each agent's utility in the Nash welfare maximizing allocation is at least half of the maximum utility they could achieve in any competitive equilibrium.
The proof relies on two key lemmas derived using Lagrangian duality and supergradient arguments. The first 
is a simple observation, showing that the Gale demand solution for lower prices yields at least half the utility of the utility maximizing bundle at the current prices.
\begin{restatable}{lemma}{galedemandapprox}\label{lem::gale-demand-approx-2}
Consider a  utility function $u\,:\, \Rp^\goods\to \Rp$, prices
$\mp,\mq \in \FQ{u}$  such that $\mathbf{q} \leq \mathbf{p}$. Then, for any $\my\in \galedemand{}^u(\mq,b)$ and $\mx\in \demand{}^u(\mathbf{p},b)$, it holds that $u(\my) \geq \frac{1}{2} u(\mx)$.
\end{restatable}
\begin{proof}
By Lemma~\ref{lem::gale-kkt},  there exists a supergradient $\mg \in \supg u(\my)$ such that $ g_j \leq q_j u(\my)/b$ with equality whenever $y_j>0$. Therefore,
    \begin{align*}
        u(\mx) &\leq  u(\my) + \pr{\mg}{\mx-\my} \leq  u(\my) + \frac{u(\my)}{b}\pr{\mq}{\mx-\my} \\
        &\le u(\my)\left(1+\frac{\pr{\mq}{\mx}}{b}\right)
        \le u(\my)\left(1+\frac{\pr{\mp}{\mx}}{b}\right)\le 2u(\my)\, . \tag*{\qedhere}
    \end{align*} 
\end{proof}

The second lemma establishes that for every allocation of goods in the standard demand correspondence, there exists a smaller or equal allocation of goods in the Gale demand correspondence. While Lemma~\ref{lem::gale-demand-approx-2} holds for any utility functions, this result critically depends on the assumption that the utility function satisfies the $\Sigma$-Gale-substitutes property. The proof is given in Appendix~\ref{sec:main-proof}, transforming the KKT-conditions of the demand correspondence $\demand{}^u(\mp,b)$ to the KKT-conditions of 
$\galedemand{}^u (\mathbf{p}, b)$  (Lemma~\ref{lem::gale-kkt}), by making use of property \eqref{prop:satiate}.

\begin{restatable}{lemma}{galedemands}\label{lem:gale-demands}
    Let  $u(\cdot)$ be a $\Sigma$-Gale-substitutes utility and let $\mx\in\demand{}^u(\mp,b)$ for $\mp\in\Rp^\goods$ and $b>0$. Then, $\mp\in\FQ{u}$, and there exists $\my\in \galedemand{}^u (\mathbf{p}, b)$ such that $ \my \leq \mx$.
\end{restatable}

We now restate Theorem~\ref{thm::main-result}, and sketch the proof. The full proof is given in Appendix~\ref{sec:main-proof}.
\mainresult*

\begin{proof}[Proof sketch]
We start with a competitive equilibrium $((\mx_i)_{i\in \agents},\mp)$, and show that the dual of the Eisenberg--Gale program \eqref{eq:Gale-dual}, $\EGdual(\mathbf{q})$, has an optimal solution $\mq^\star$ such that $\mq^\star\le \mp$. The statement then follows from Lemma~\ref{lem::gale-demand-approx-2}.

An intuitive argument is as follows: starting from $\mq=\mp$, where the total Gale demand is less than the supply (by Lemma~\ref{lem:gale-demands}), we decrease the prices in a t\^atonnement-like procedure~\cite{cheung2020tatonnement} using the Gale-substitutes property.
As long as $\mq$ is not a minimizer of $\EGdual(.)$, there will be an underdemanded good $j$ with respect to Gale demands. Reducing $q_j$ for such a good $j$ may only decrease the Gale demand for all agents on all other goods. This iterative process converges to an equilibrium point and, since the prices are only decreasing, we can conclude that the optimal limit price vector satisfies $\mq^\star \leq \mathbf{p}$.

Rather than analysing this continuous procedure, the proof in Appendix~\ref{sec:main-proof} establishes $\mq^\star \leq \mathbf{p}$ directly, using Lagrangian duality. For any price vector $\mq$, we get a subgradient in $\subg\varphi(\mathbf{q})$ from the difference between the supply and the total Gale demand.  Further, for $\mq \leq \mathbf{p}$ we get that if $q_j = p_j$ then the $j$-th component of this subgradient at $\mq$ is non-negative. This is shown in Lemma~\ref{lem::bounds-aggregate}, by making use of an allocation $\my$ derived using Lemma~\ref{lem:gale-demands} for the prices $\mq$. Thus, in the region $\{\mq ~|~ \0 \leq \mq \leq \mathbf{p} \}$, the subgradients on the boundary all point outward. Consequently, if $\mq^\star\not\leq \mp$, then we can take the projection $\mq$ of $\mq^\star$ to the boundary of the region, where this subgradient shows that $\mq$ is also optimal.
\end{proof}

The following example shows the bound in Theorem~\ref{thm::main-result} is tight.
\begin{example}
Consider a market with $n$ agents and $n$ goods, where all agents have bounded linear utilities. Agent 1's utility is
\[u_1(\mx_1) = \sum_{j} \min\{1, x_{1j}\},\]
while, for all other agents ($2 \leq i \leq n$), they all have identical preferences on all goods except for the first one
\[u_i(\mx_i) = \sum_{j > 1} \min\{1, x_{ij}\}.\]
In the Nash welfare maximizing allocation, each agent receives a bundle yielding utility exactly equal to $1$.
However, consider a competitive equilibrium in which the price of the first good is $0$, and the prices of all other goods are set to $\frac{n}{n - 1}$. In this equilibrium, agent $1$ obtains the entire first good and a $\frac{n - 1}{n}$ fraction of the other goods, leading to a total utility of
\[
    u_1 = 1 + \frac{n - 1}{n},
\]
which approaches $2$ as $n \to \infty$.
\end{example}

\begin{remark}
The assumption that the utilities are   $\Sigma$-Gale-substitutes is important in Theorem~\ref{thm::main-result}.  Example \ref{example:non-gale} demonstrates that for non-Gale substitute utilities, some agents may achieve significantly higher utility in a competitive equilibrium than in the Nash welfare maximizing allocation.
\end{remark}

\begin{remark}
Theorem~\ref{thm::main-result} naturally raises a question about the converse: for any agent $i$, does there always exist a CE $(\mx_i)_{i\in \agents}$ such that $u_i(\mx_i) \geq C \cdot u_i(\my_i)$ for some constant $C > 0$, where $(\my_i)_{i\in \agents}$ is a Nash welfare maximizing allocation? We show that such a constant-factor guarantee is not always possible by presenting an example (see Example~\ref{example::E4}) in which the CE is unique, yet there exists an agent whose utility under this equilibrium is significantly lower than in the Nash welfare maximizing allocation.
\end{remark}

\section{Gale substitute utilities}\label{sec:gale_substitute_utilities}
First, we show that separable utilities are $\Sigma$-Gale-substitutes.
\sepgale*
For this, we first adapt Lemma~\ref{lem::gale-kkt} to the setting of separable utilities.
\begin{lemma}\label{lem::kkt-separable}
Let  $u(\mx)=\sum_{j\in\goods} v_j(x_j)$ be a separable utility function.
For any $b > 0$, $\mathbf{q} \in \FQ{u}$, $\my \in \galedemand{}^u (\mq,b)$, we have $u(\my) > 0$. 
Further,  $\my \in \galedemand{}^u (\mathbf{q}, b)$ if and only if there exist  supergradients $g_{j} \in \supg v_{j} (y_{j})$, $j\in\goods$ satisfying
  \begin{enumerate}[(i)]
        \item $g_{j} \geq 0$; and 
        \item $g_{j}\leq q_j u(\my)/b$ with equality whenever $y_{j} > 0$.
    \end{enumerate}
\end{lemma}
We now show that strictly concave and differentiable utility functions satisfy Definition \ref{def:gale-subs} \eqref{def:gale-subs-1} : for any  $b > 0$, any $\mathbf{q},\mq' \in \FQ{u}$, and any $\my \in \galedemand{}^u (\mq,b)$, if  $\mathbf{q} \geq \mathbf{q}'$, then there exists a $\my' \in \galedemand{}^u (\mathbf{q}',b)$ such that $y'_{j} \leq y_{j}$ whenever $q_j' = q_j$.
\begin{proof}[Proof of Theorem~\ref{thm:sep-gale} for strictly concave and differentiable utilities] 
  Let $u(\my)=\sum_{j\in\goods} v_j(y_j)$ be a separable and differentiable concave utility function, and let $\nabla_{\!j} u(\my)=\partial v_j(y_j)/\partial y_j$ denote the $j$-th component of the gradient. 
  Recall that for differentiable utilities, the only supergradient of $v_j(y_j)$ is $\nabla_{\!j} u(\my)$. By 
  Lemma~\ref{lem::kkt-separable}, 
 $\nabla_{\!j} u(\my)\le q_j u(\my)/b$, with equality whenever $y_j>0$. Let $\my\in \galedemand{}^u(\mq,b)$ and $\my'\in \galedemand{}^u(\mq',b)$ for $\mq'\le \mq$. 

We first claim that $u(\my')\ge u(\my)$. For a contradiction, assume $u(\my')< u(\my)$.
 From the above, whenever $y_j>0$, we have  
\[
  \nabla_{\!j} u(\my)=q_j u(\my)/b\ge q'_j u(\my')/b\ge \nabla_{\!j} u(\my')\, .   
\]
 By strict concavity, this means $y'_j\ge y_j$. Thus, we see that $\my'\ge \my$, a contradiction to the monotonicity of $u(.)$ and the assumption $u(\my')< u(\my)$.

From $u(\my')\ge u(\my)$,  for every good $j$ with $q_j=q'_j$ and $y_j>0$,  $\nabla_{\!j} u(\my')\ge \nabla_{\!j} u(\my)$ follows. By strict concavity, we must have $y'_j\le y_j$ as required. 
\end{proof}
The proof without the differentiability and strict concavity assumptions poses additional difficulties: for example, we would not be able to conclude $y'_j\le y_j$ in the last step, and might need a specific choice of $y'_j$ from  $\galedemand{}^u(\mq',b)$. The proof of the general case is given in Appendix~\ref{sec:separable}.

\medskip

Our final theorem is the Gale substitutability of generalized network utility functions. This is deferred to Appendix~\ref{sec:network}, and we only present a high level summary here.
\networkgale*

To verify the Gale substitutes property, we need to take two price vectors $\mq\ge \mq'$ and $\my\in \galedemand{}^u (\mathbf{q}, b)$, and construct a $\my'\in \galedemand{}^u (\mathbf{q'}, b)$ such that $y'_j\le y_j$ whenever $q'_j=q_j$; we have a second, analogous construction for budget changes. The proof is a finite algorithm that transforms the optimal generalized flow for demands $\my$ to $\my'$ such that we throughout maintain the primal-dual slackness condition
corresponding to the Gale demand correspondence. The algorithm builds on the literature on maximum generalized flow algorithms. In particular, we use ideas from the classical highest gain augmenting path algorithm by Onaga~\cite{Onaga66}. While this does not give a polynomial-time algorithm, it was the starting point of further such algorithms, including \cite{Goldberg1991,olver2020simpler,Vegh11}. 

Although generalized network utilities are not $\Sigma$-Gale-substitutes in general (see Example~\ref{exp:leontief-free-satiated}), we can still establish the following theorem with a slight modification to the proof of Theorem~\ref{thm::main-result}; the proof is also given in Appendix~\ref{sec:network}.

 \networkgalebound*

\section{Discussion}\label{sec:discussion}
In this paper, we compared competitive equilibrium and Nash welfare maximizing solutions in the Fisher market setting. We identified Gale-substitutes as a broad family of utility functions where the Nash welfare uniformly guarantees every agent at least half of their maximum utility at any competitive equilibrium. We also introduced generalized network utilities. This is a very general construction with  rich combinatorial structure and may be a relevant example in various market and fair division settings.

Gale-substitutes may form a compelling subclass of concave functions worth investigating on its own. In the context of indivisible goods, the analogous class of gross substitute utilities, introduced by Kelso and Crawford~(\citeyear{kelso1982job}), turned out to be the same as the important class of $M^{\natural}$-concave utilities in discrete convex analysis \cite{fujishige2003note}. 
In particular, one may wonder whether the Gale-substitutes property implies \emph{submodularity}, that is, $u(\mx)+u(\my)\ge u(\mx\vee \my)+u(\mx\wedge \my)$ for any $\mx,\my\in\Rp^\goods$, where $\mx\vee \my$ and $\mx\wedge \my$ denote the pointwise maximum and minimum of the two vectors, respectively. This property trivially holds for separable utilities and can also be shown for generalized network utilities. Another interesting question is whether computing competitive equilibrium in generalized network utilities lies in PPAD.

We note that one can construct an example of a concave quadratic function that is submodular but not Gale-substitutes. On the other hand, it is unclear whether Theorem~\ref{thm::satiate-gale} would remain true for submodular utilities instead of Gale-substitutes.  One would also be interested in whether Theorem~\ref{thm::satiate-gale} holds for matching utilities, which, like submodular utilities, are not necessarily Gale-substitutes (see Example~\ref{exp:matching}). Another natural question related to the well-studied notion of \emph{weak gross substitutes (WGS)} utilities: Gale-substitutes includes utility functions that are not WGS, such as concave separable; but it is not clear whether all WGS utilities are also Gale-substitutes. 

Gale-substitutes utilities may also have potential algorithmic applications. The WGS property leads to natural t\^atonnement and auction mechanisms for finding competitive equilibria, see, e.g., \cite{ArrowBH,ArrowH60, bei2019ascending,codenotti2005market,codenotti2005polynomial, garg2023auction}. In particular, \cite{garg2023auction} gives an auction algorithm for maximizing Nash welfare under \emph{capped} SPLC utilities. Similar methods could be extended to the broader class of Gale substitutes utilities. The Gale-substitutes property may be useful to design simple such mechanisms and combinatorial algorithms for solving the convex program \eqref{def:log-NSW}. For the linear case, Devanur et al.~(\citeyear{DevanurPSV08}) gave a primal-dual algorithm that was followed by a series of works on combinatorial algorithms for various market settings; see, e.g.,~\cite{ChaudhuryGMM22,ColeG15,DuanM15,GargV19}. 

\appendix
\section{Properties of the Gale demand correspondence}\label{sec:galedemand}
In this Appendix, we prove Lemma~\ref{lem-galedbasic}. We also derive some technical properties of the price domain $\FQ{u}$ that will be used in the subsequent Appendices.

\galedbasic*
\begin{proof}
Let $M \defeq \sup_{\my\in\Rp^\goods} b \log u(\my)-\pr{\mp}{\my}$, 
which implies the existence of a sequence of allocations $\{\my^{(k)}\}_{k = 1}^\infty$ such that $b \log u(\my^{(k)})-\pr{\mp}{\my^{(k)}}$ converges to $M$. We note that the finiteness of $M$ follows from Lemma~\ref{lem::budget-spending-NSW}, which is explicitly established in Lemma~\ref{lem::gale-domain-property}(i).
We show the existence of an allocation $\my^*$ satisfying $b \log u(\my^*)-\pr{\mp}{\my^*} = M$. 

According to Lemma~\ref{lem::budget-spending-NSW}, there exists a sequence of $\bar{\my}^{(k)}$ such that $b \log u(\bar{\my}^{(k)})-\pr{\mp}{\bar{\my}^{(k)}} $ converges to $M$, with $0 \leq \pr{\mp}{\bar{\my}^{(k)}} \leq b$ for all $k$. Thus, we can select a subsequence $\bar{\my}^{(i_k)}$ such that $b \log u(\bar{\my}^{(i_k)})-\pr{\mp}{\bar{\my}^{(i_k)}}$ converges to $M$ and $ \pr{\mp}{\bar{\my}^{(i_k)}}$ converges to $b'$ for some $b' \in [0, b]$. Consequently, $b \log u(\bar{\my}^{(i_k)})$ converges to $M + b'$ as $ \pr{\mp}{\bar{\my}^{(i_k)}}$ converges to $b'$. By Assumption~\ref{assum:basic-demand-system}, there exists an allocation $\my^*$ such that $b \log u(\my^*) = M + b'$ and $\pr{\mp}{\my^*} \leq b'$. This completes the proof.
\end{proof}

\begin{restatable}{lemma}{galedomainproperty}\label{lem::gale-domain-property}
For any utility function $u\,:\, \Rp^\goods\to \Rp$ and $b > 0$, the following properties of the price domain $\FQ{u}$ hold:
\begin{enumerate}[(i)]
\setlength\itemsep{0em}
\item for any $\mathbf{q} \in \FQ{u}$, $\max_{\my\in\Rp^\goods} b \log u(\my) - \pr{\mathbf{q}}{\my} < +\infty$; \label{lem::arb-budget} 
\item $\{ \mathbf{q} ~|~ \mathbf{q} > 0 \} \subseteq \FQ{u}$; \label{lem::gale-demand-positive-q}
\item for any $\mathbf{q} \geq \mathbf{q}' \geq 0$ such that $\{j : q'_j = 0\} = \{j : q_j = 0\}$, $\mq \in \FQ{u}$ implies $\mq' \in \FQ{u}$;  \label{lem::feasi-q-zero}
\item for any $\mathbf{p} \geq 0$, if $\max_{\pr{\mp}{\my} \leq b} \log u(\my) < +\infty$, then $\mathbf{p} \in \FQ{u}$.  \label{lem::demand-gale-finite}
\item for any $M>0$, if there is a convergent sequence of prices $\mp^{(k)}\in\FQ{u}$ with  $M > \max_{\my\in\Rp^\goods} b \log u(\my)-\pr{\mp^{(k)}}{\my}$, $k=1,2,\ldots$ with a limit point
 $\mp^{(k)}\to \bar \mp$, then $\bar\mp\in \FQ{u}$.\label{part:limit-point}
\end{enumerate}
\end{restatable}
\begin{proof}
    We first prove property \eqref{lem::arb-budget}. According to the definition of $\FQ{u}$, there exists $b'>0$ such that $\max_{\my\in\Rp^\goods} b' \log u(\my) - \pr{\mathbf{q}}{\my}  < +\infty$. 
    For any $b>0$ and for any $\my\in\Rp^\goods$, if $\pr{\mathbf{q}}{\my} \leq b$, then  $b \log u(\my) - \pr{\mathbf{q}}{\my} <  \frac{b}{b'} \left( b' \log u(\my) -  \pr{\mathbf{q}}{\my}  \right) + \frac{b - b'}{b'} \ \pr{\mathbf{q}}{\my} $, which is bounded. Then, the result follows using Lemma~\ref{lem::budget-spending-NSW}.

    Property \eqref{lem::gale-demand-positive-q} holds as, by Lemma~\ref{lem::budget-spending-NSW}, $\max_{\my\in\Rp^\goods} b \log u(\my) - \pr{\mq}{\my} \leq \max_{\pr{\mq}{\mx} \leq b}  b \log u(\mx) -\pr{\mq}{\mx}$, which is bounded as $\mx$ is bounded: $x_j \leq \frac{b}{q_j} < \infty$.

    To prove property \eqref{lem::feasi-q-zero}, we let $\alpha = \min_{q_j > 0} \{ q'_j  / q_j \}$. Then,  $\max_{\my\in\Rp^\goods} (\alpha b) \log u(\my) -  \pr{\mq'}{\my}\leq \max_{\my\in\Rp^\goods} (\alpha b) \log u(\my) -  \alpha \pr{\mq}{\my}< + \infty$. 
    
    Property \eqref{lem::demand-gale-finite} holds as,  for any $\my$ such that $\pr{\mp}{\my} \leq b$, $b \log u(\my) - \pr{\mp}{\my} \leq b \log u(\my) < +\infty$. 

Finally, let us prove \eqref{part:limit-point}. For a contradiction, assume $\bar\mp\notin\FQ{u}$, that is, $\sup_{\my\in\Rp^\goods} b \log u(\my)-\pr{\bar\mp}{\my}=\infty$. Consider a vector $\my\in\Rp^\goods$ with $b \log u(\my)-\pr{\bar\mp}{\my}>2M$. Clearly, $\my\neq 0$. Let us select $k$ such that $\|\mp^{(k)}-\bar\mp\|<M/\|\my\|$. Then, $b \log u(\my)-\pr{\mp^{(k)}}{\my}= b \log u(\my)-\pr{\bar\mp}{\my}+\pr{\bar\mp-\mp^{(k)}}{\my}>M$, in contradiction with the assumption on the sequence $\mp^{(k)}$.
\end{proof}

\section{Efficiency of Nash welfare for Gale-substitutes}\label{sec:main-proof}
In this Appendix, we give the full proof of Theorem~\ref{thm::main-result}.
We show that for $\Sigma$-Gale substitutes utility functions, any competitive equilibrium $((\mx_i)_{i\in \agents},\mp)$ admits Gale prices $\mq^\star$ such that $\mq^\star\le \mp$. The theorem statement then follows from Lemma~\ref{lem::gale-demand-approx-2}. 
We start by re-stating and proving Lemma~\ref{lem:gale-demands}.

\galedemands*
\begin{proof}
    According to the KKT conditions, there exists $\lambda \geq 0$ and supergradient $\mg \in \supg u(\mx)$ such that
    \begin{equation}\label{eq:kkt-gale} g_{j} \leq \lambda p_j\, , \quad\forall j\in \goods\,,\,\mbox{with equality  if }x_{j} > 0\, ,\end{equation}
and $\pr{\mp}{\mx}=b$  if $\lambda>0$.

  Let us first consider the case $\lambda=0$. Then, $g_j\le 0$ with equality if $x_j>0$, implying that $\mx=\arg\max_{\mx\in\Rp^\goods} u(\mx)$. Consequently, $u(.)$ is satiable; therefore, $\mathbf{0} \in \FQ{u}$ and $\mx\in \galedemand{}^u (\0, b)$. Now, \eqref{prop:satiate} guarantees the existence of $\my\in \galedemand{}^u (\mathbf{p}, b)$, $\my\leq \mx$. Clearly, $\FQ{u}=\Rp^\goods$ in this case.

Next, assume $\lambda>0$.
Let us define 
\[
\mq\defeq \frac{\lambda b}{u(\mx)}\cdot \mp\, .
\]
Note that $\mq\le\mp$. This follows since $\lambda b = \lambda \pr{\mp}{\mx}=\pr{\mg}{\mx}\le u(\mx)$, where the final inequality is by concavity and $u(\0)=0$.
Thus,
    $$ g_{j} \leq \frac{u(\mx)}{b}q_j\, , \quad\forall j\in \goods\,,\,\mbox{with equality  if }x_{j} > 0\, .$$
    According to Lemma~\ref{lem::gale-kkt},  $\mx \in \galedemand{}^u (\mathbf{q}, b)$.

   Note that $\galedemand{}^u (\mathbf{q}, b) = \galedemand{}^u (\mp, u(\mx) / \lambda)$.  Utilizing Definition~\ref{def:gale-subs}\eqref{def:gale-subs-2} for $\mx \in \galedemand{}^u (\mp, u(\mx) / \lambda)$ and $b\le u(\mx)/\lambda$, there exists $\my \in \galedemand{}^u (\mathbf{p}, b)$ such that $\my \leq \mx$ as required.
\end{proof}
We next show that 
for any competitive equilibrium $((\mx_i)_{i\in \agents}, \mathbf{p})$,
and any prices $\mq\le\mp$, there exist an aggregate Gale demand such that 
$\sum_{i\in\agents} y_{ij}\le 1$ whenever $q_j=p_j$.
\begin{lemma} \label{lem::bounds-aggregate}
Consider a Fisher market instance such that 
$u_i(\cdot)$ is $\Sigma$-Gale-substitute utilities  for all agents. For any competitive equilibrium $((\mx_i)_{i\in \agents}, \mathbf{p})$, for any $\mathbf{q}  \in \cap_{i\in\agents} \FQ{u_i}$ such that $\mathbf{q} \leq \mathbf{p}$, there exists $\my_i \in \galedemand{i} (\mathbf{q}, b_i)$ for all $i\in\agents$ and $\sum_{i\in \agents} y_{ij} \leq 1$ for all $j\in \goods$ such that $q_j = p_j$.
\end{lemma}
\begin{proof}
By Lemma~\ref{lem:gale-demands}, there exists $\bar \my_i\in \galedemand{i} (\mathbf{p}, b_i)$ for every $i\in\agents$ such that $\bar \my_i\le \mx_i$. Using Definition~\ref{def:gale-subs}\eqref{def:gale-subs-1}, there exists  $\my_i\in \galedemand{i} (\mathbf{q}, b_i)$ and $y_{ij}\leq \bar y_{ij}$ if $q_j=p_j$. Thus, if $q_j=p_j$ then $\sum_{i\in \agents} y_{ij} \leq \sum_{i\in \agents} \bar y_{ij}\le \sum_{i\in \agents} x_{ij}\le 1$.
\end{proof}
Recall the Lagrangian dual function $\EGdual$ defined in \eqref{eq:Gale-dual-function} on the domain $\dom(\EGdual)=\cap_{i\in\agents} \FQ{u_i}$.
\begin{align*}
  \EGdual (\mq) \defeq \sum_{i\in\agents} \max_{\my_i\in\Rp^{\goods}}\left(b_i \log u_i(\my_i) - \pr{\mq}{\my_i}\right) + \sum_{j\in \goods} q_j.
\end{align*}
The properties stated in the next lemma are easy to verify.
\begin{lemma}\label{fact::subgrad-NSW}The following hold for the function $\EGdual (\cdot)$.
\begin{enumerate}[(i)]
\setlength\itemsep{0em}
    \item $\EGdual (\mathbf{q})$ is convex.
    \item Given $\mq\in \dom(\EGdual)$, let
    $\my_i \in \galedemand{i}(\mathbf{q}, b_i)$ for all $i\in \agents$. Then, $\mathbf{1} - \sum_{i\in\agents} \my_i \in \subg \EGdual (\mathbf{q})$.
    \item The optimum value of \eqref{eq:Gale-dual}, that is, the infimum of  $\EGdual (\mathbf{q})$ is taken in $\dom(\EGdual)$. 
\end{enumerate}
\end{lemma}

\begin{proof}[Proof of Theorem~\ref{thm::main-result}]
Let $((\mx_i)_{i\in \agents},\mp)$ be a competitive equilibrium solution. Consider the dual \eqref{eq:Gale-dual} of 
 the Eisenberg-Gale program \eqref{def:log-NSW}, i.e., 
\[
\min_{\mq\in\dom(\EGdual)} \EGdual(\mq)\, .
\]
We show that this has an optimal solution with $\mq\le \mp$. The theorem then follows by Lemma~\ref{lem::gale-demand-approx-2}. 
Let us take any optimal $\mq^\star\in\dom(\EGdual)$; we are done if $\mq^\star\le\mp$. Otherwise, let $J\defeq\{j\in \goods\, :\, q^\star_j>p_j\}$, and let us define $\bar\mq\in\Rp^\goods$ as
\[
\bar q_j\defeq\begin{cases} p_j\, ,&\mbox{if }j\in J\, ,\\
q^\star_j\, ,&\mbox{if }j\notin J\, .\\
\end{cases}
\]
First, assuming $\bar\mq\in\dom(\EGdual)$, we show that $\bar \mq$ is also optimal to \eqref{eq:Gale-dual} as claimed.

Let us pick an allocation $(\my_i)_{i\in \agents}$ for the prices $\bar\mq$ as guaranteed by Lemma~\ref{lem::bounds-aggregate}: $\sum_{i\in\agents} y_{ij}\le 1$ for all $j\in\goods$ with $\bar q_j=p_j$; in particular, for all $j\in J$. According to Lemma~\ref{fact::subgrad-NSW}(ii), we get a subgradient 
$\bar\mg\in\subg\EGdual(\bar\mq)$ defined as $\bar\mg\defeq\mathbf{1} - \sum_{i\in\agents} \my_i$; this has the property $\bar g_j\ge 0$ for all $j\in J$.
We then have
\[
\EGdual(\mq^\star)\ge \EGdual(\bar\mq)+\pr{\bar \mg}{\mq^\star-\bar\mq}=\EGdual(\bar\mq)+
\sum_{j\in J} \bar g_j (q^\star_j-\bar q_j)\ge \EGdual(\bar\mq)\, .
\]
Here, the first inequality is by the definition of subgradients; the equality uses that $q^\star_j=\bar q_j$ if $j\notin J$, and the final inequality follows since $\bar g_j\ge 0$ and $q^\star_j>\bar q_j$ for $j\in J$. This 
shows that $\bar\mq$ is also an optimal solution to \eqref{eq:Gale-dual},  completing the proof assuming $\bar\mq\in\dom(\EGdual)$.

It is left to show that $\bar\mq\in\dom(\EGdual)$. Let  $K\defeq\{j\in \goods\, :\, q^\star_j=0<p_j\}$. Clearly, $K\cap J=\emptyset$. If $K=\emptyset$, then Lemma~\ref{lem::gale-domain-property}\eqref{lem::feasi-q-zero} already implies $\bar\mq\in\FQ{u}$, since $\mp\ge \bar\mq$ and $\supp(\bar\mq)=\supp(\mp)$.  Otherwise, for $k=1,2,\ldots$, let us define
$\mq^{(k)}\in\Rp^\goods$ as
\[
q_j^{(k)}\defeq\begin{cases} p_j\, ,&\mbox{if }j\in J\, ,\\
\frac{p_j}{k}\, ,&\mbox{if }j\in K\, ,\\
q^\star_j\, , &\mbox{if }j\notin J\cup K\, .\\
\end{cases}
\]
We have $\mq^{(k)}\in \FQ{u}$ again by Lemma~\ref{lem::gale-domain-property}\eqref{lem::feasi-q-zero}. As in the first part, we take a subgradient $\mg^{(k)}$ as guaranteed by Lemma~\ref{fact::subgrad-NSW}(ii), with the properties that $g^{(k)}_j\ge 0$ for all $j\in J$ and also $\mg^{(k)}\le\mathbf{1}$.
We can write
\[
\EGdual(\mq^\star)\ge \EGdual(\mq^{(k)})+\pr{\mg^{(k)}}{\mq^\star-\mq^{(k)}}=\EGdual(\mq^{(k)})+
\sum_{j\in J} g^{(k)}_j (q^\star_j- q^{(k)}_j)-\sum_{j\in K} g^{(k)}_j q^{(k)}_j\ge \EGdual(\mq^{(k)})-\frac{\|\mp\|_\infty}{k} |K|\, .
\]
Here, we use in the first equality that $q^\star_j=q^{(k)}_j$ for $j\notin J\cup K$, and $q^\star_j=0$ for $j\in K$. The last inequality uses $g^{(k)}_j \le 1$
and the definition $q^{(k)}_j=p_j/k$. Thus,
 Lemma~\ref{lem::gale-domain-property}\eqref{part:limit-point} is applicable with $M=\EGdual(\mq^\star)+|K|\cdot\|\mp\|_\infty$, using that for each $j\in\goods$, $\max_{\my_i\in\Rp^\goods}b_i\log u_i(\my_i)-\pr{\mq^{(k)}}{\my_i}\le \EGdual(\mq^{(k)})$. This  shows that $\bar \mq\in\FQ{u_i}$ for all $i\in \goods$, and therefore $\bar\mq\in\dom(\EGdual)=\cap_{i\in\goods}\FQ{u_i}$.
\end{proof}

\section{Separable utilities are $\Sigma$-Gale-substitutes}\label{sec:separable}
In this Appendix, we prove Theorem~\ref{thm:sep-gale} for the general case.
Recall the following simple monotonicity property of supergradients of univariate concave functions.
\begin{lemma}\label{lem:univariate}
Let $v\,:\,\R\to\R$ be a univariate concave function, $\alpha,\alpha'\in\dom(v)$, and $g\in\supg v(\alpha)$, $g'\in\supg v(\alpha')$. Then, $\alpha<\alpha'$ implies $g\ge g'$. Conversely, if $g> g'$, then $\alpha\le \alpha'$.
\end{lemma}

The proof of Theorem~\ref{thm:sep-gale} follows by the next three lemmas, that prove the three properties in Definition~\ref{def:gale-subs} for separable utilities: property  \eqref{def:gale-subs-1} in Lemma~\ref{lem::sep-def-1}, property \eqref{def:gale-subs-2} in Lemma~\ref{lem::sep-def-2}, and property~\eqref{prop:satiate} in Lemma~\ref{lem::sep-def-3}. A similar argument was used in Lemma 7.1 in \cite{garg2023auction-arxiv} to show a variant of the Gale substitutes property for capped SPLC utilities.
\begin{lemma}\label{lem::sep-def-1}
     Let  $u(\mx)=\sum_{j\in\goods} v_j(x_j)$  be a separable utility function, and let $b > 0$, $\mq,\mq' \in \FQ{u}$   and  $\my \in \galedemand{}^u (\mathbf{q}, b)$. If  $\mathbf{q} \geq \mathbf{q}'$, then there exists a $\my' \in \galedemand{}^u (\mathbf{q}', b)$ such that $y'_{j} \leq y_{j}$ whenever $q_j' = q_j$.
\end{lemma} 
\begin{proof}
Let $J\defeq\{j\in\goods\, :\, q'_j=q_j \}$. 
Let us select  $\my' \in \galedemand{}^u (\mathbf{q}', b)$  minimizing the potential $\Psi(\my')\defeq\sum_{j\in J} (y'_{j} - y_{j})^+$. Our goal is to show $\Psi(\my')=0$.
Throughout the proof, let $g'_{j} \in \supg v_{j}(y'_{j})$ and $g_{j} \in \supg v_{j}(y_{j})$ as in Lemma~\ref{lem::kkt-separable}.

First, we show $u(\my')\ge u(\my)$.
For a contradiction, assume $u(\my')<u(\my)$. Then, according to Lemma~\ref{lem::kkt-separable}, for every $j\in\goods$ with $y_j>0$ and $q_j>0$ we have
$$g'_{j}\le q'_j  u(\my')/b < q_j u(\my)/b= g_j  \, ,$$
while for $j\in\goods$ with $q_j=0$ we have $g'_j=g_j=0$. Consequently, for each $y_j>0$, if $q_j>0$ then Lemma~\ref{lem:univariate} implies $y'_j\ge y_j$, and if $q_j=0$ then $v_j(y_j)=v_j(y_j')$, both being equal to the maximum value of $v_j(.)$. These imply $u(\my')\ge u(\my)$, a contradiction.

 Using the above, 
 for any $j\in J$ where $y'_j>0$, 
\begin{equation}
 g'_{j}  = \frac{q'_j}{b} u(\my') \ge \frac{q_j}{b} u(\my) \ge g_j \, . \numberthis \label{eqn::separable-eqn-1}
\end{equation}
Let us first consider the case $u(\my')>u(\my)$.
For all $j\in J$ with $q'_j=q_j>0$, we get $g'_j>g_j$.  By Lemma~\ref{lem:univariate}, $y'_j\le y_j$ follows for such goods $j$. On the other hand,  if  $q'_j = q_j = 0$, then this implies $g'_j=g_j=0$. For such indices, if $y'_j>y_j$, then replacing $y'_j$ by $y_j$ results in another bundle $\my''\in\galedemand{}^u(\mq',b)$ 
 with $\Psi(\my'')<\Psi(\my')$, a contradiction.

For the rest of the proof, let us assume that $u(\my')=u(\my)$.
We next show $y'_j\ge y_j$ for all $j\in\goods\setminus J$, that is, for all goods with $q'_j<q_j$. This is clearly true if $y_j=0$. Whenever $y_j>0$ for such an index, we have 
\[
 g'_{j}  \le \frac{q'_j}{b} u(\my') < \frac{q_j}{b} u(\my) = g_j \, ,
\]
and therefore $y'_j\ge y_j$ by Lemma~\ref{lem:univariate}. 
By monotonicity, $v_j(y'_j)\ge v_j(y_j)$ for all $j\in\goods\setminus J$. Since $u(\my')=u(\my)$, we have $\sum_{j\in J} v_j(y'_j)\le \sum_{j\in J} v_j(y_j)$.

From \eqref{eqn::separable-eqn-1}, it follows that $g'_j=g_j$ for all $j\in J$ with $y'_j>0$. This is because the first inequality is now equality, and the second inequality may only be strict if $y_j=0<y'_j$. By Lemma~\ref{lem:univariate}, $g_j\ge g'_j$ in this case, leading to a contradiction. It also follows that $g'_k = g_k$ for all $k \in J$ with $y_k > 0$ from the same argument.

Assume for a contradiction 
$\Psi(\my')>0$, and
let us select an index $j\in J$ with $y'_j>y_j$. If $v_j(y'_j)=v_j(y_j)$, then we get a contradiction, because changing $y'_j$ to $y_j$ would preserve the optimality conditions in Lemma~\ref{lem::kkt-separable}, while decreasing the potential $\Psi(\my')$.
Thus, $v_j(y'_j)>v_j(y_j)$. By $\sum_{j'\in J} v_{j'}(y'_{j'})\le \sum_{j'\in J} v_{j'}(y_{j'})$, there must be an index $k\in J$ with  $v_k(y'_k)<v_k(y_k)$. We note as $u(\my) = u(\my')$ and $j \in J$ and $k \in J$, this implies $g'_k = g_{k} = \frac{q'_k}{b} u(\my')$ and $g'_j =  g_j = \frac{q'_j}{b} u(\my')$. We can now find small $\varepsilon_k,\varepsilon_j>0$  such that $\varepsilon_k \cdot q'_k = \varepsilon_j \cdot q'_j$; $y'_j - \varepsilon_j\ge y_j$ and $y'_k + \varepsilon_k \le y_k$. We note after making the update $y''_j = y'_j - \varepsilon_j\ge y_j$ and $y''_k = y'_k + \varepsilon_k\le y_k$. Additionally, $v_k(y'_k + \varepsilon_k)+v_j(y'_j - \varepsilon_j)=v_k(y'_k)+v_j(y'_j)$ and $y''_j q'_j + y''_k q'_k = y'_j q'_j + y'_k q'_k$. Therefore, this modified solution $\my''$ satisfies the optimality conditions but $\Psi(\my'')<\Psi(\my')$, again leading to a contradiction.
\end{proof}

\begin{restatable}{lemma}{sepbudgetdec}\label{lem::sep-def-2}
  Let  $u(\mx)=\sum_{j\in\goods} v_j(x_j)$ be a separable utility function, $b > b' > 0$ and $\mathbf{q} \in \FQ{u}$. Then,  for any $\my \in \galedemand{}^u (\mathbf{q}, b)$, there exists a $\my' \in \galedemand{}^u (\mathbf{q}, b')$ such that $\my' \leq \my$.
\end{restatable}
\begin{proof}
We argue similarly as in the previous proof: we pick $\my' \in \galedemand{}^u (\mathbf{q}, b')$ which minimizes $\Psi(\my')\defeq\sum_{j\in \goods} (y'_{j} - y_{j})^+$, and show that $\Psi(\my')=0$.

First, we show that ${u(\my)}/{b} \leq {u(\my')}/{b'}$. For a contradiction, suppose ${u(\my)}/{b} > {u(\my')}/{b'}$. If $y_{j} > 0$, then, by Lemma~\ref{lem::kkt-separable} 
$$g_{j} = \frac{q_j u(\my)}{b} > \frac{q_j u(\my')}{b'} \geq g'_{j},$$ which implies $v_{j}(y_{j}) \leq v_{j}(y'_{j})$ whenever $y_j>0$ (Lemma~\ref{lem:univariate}), and therefore, $u(\my) \leq u(\my')$. 
Together with $b>b'$, this implies ${u(\my)}/{b} \leq {u(\my')}/{b'}$, a contradiction. 

Whenever $y'_j>0$, we now have 
\begin{equation}\label{eqn::separable-eqn-2}
g'_{j} = \frac{q_j u(\my')}{b'}  \geq  \frac{q_j u(\my)}{b}\ge g_{j} \, .
\end{equation}
Assume first $u(\my')/b'>u(\my)/b$. If $q_j>0$, then inequality \eqref{eqn::separable-eqn-2} is strict, and therefore $y'_j\le y_j$ for all $j\in\goods$ follows by Lemma~\ref{lem:univariate}. If $q_j=0$, then $g'_j=g_j=0$. If $y'_j>y_j$ for such a good $j$, we can decrease $y'_j$ to $y_j$ and thereby get an optimal allocation with strictly smaller $\Psi(\my')$ value, leading to a contradiction. This completes the proof for $u(\my')/b'>u(\my)/b$.

For the rest, we may assume $u(\my')/b'=u(\my)/b$. Now, $g'_j=g_j$ must hold whenever $y'_j>0$. This is because the first inequality in \eqref{eqn::separable-eqn-2} is equality, and if the second is strict, then $y_j=0<y'_j$ would give a contradiction to $g'_j>g_j$. It also follows from the same argument that $g'_k = g_k$ for all $k \in \goods$ with $y_k > 0$.

Assume for a contradiction 
that $\Psi(\my')>0$, and
let us select an index $j\in \goods$ such that $y'_j>y_j$. If $v_j(y'_j)=v_j(y_j)$, we obtain a contradiction because changing $y'_j$ to $y_j$ would preserve the optimality conditions in Lemma~\ref{lem::kkt-separable} while decreasing the potential $\Psi(\my')$. Hence, we must have $v_j(y'_j)>v_j(y_j)$. Since $\sum_{j'\in \goods} v_{j'}(y'_{j'})\le \sum_{j'\in \goods} v_{j'}(y_{j'})$, there must exist an index $k\in \goods$ with $v_k(y'_k)<v_k(y_k)$. Noting that $u(\my) = u(\my')$ and $j, k \in \goods$, this implies $g'_k = g_{k} = \frac{q'_k}{b} u(\my')$ and $g'_j =  g_j = \frac{q'_j}{b} u(\my')$. We can now choose small $\varepsilon_k,\varepsilon_j>0$  such that $\varepsilon_k \cdot q'_k = \varepsilon_j \cdot q'_j$; $y'_j - \varepsilon_j\ge y_j$ and $y'_k + \varepsilon_k \le y_k$. After making this update, define $y''_j = y'_j - \varepsilon_j\ge y_j$ and $y''_k = y'_k + \varepsilon_k\le y_k$. Moreover, we have $v_k(y'_k + \varepsilon_k)+v_j(y'_j - \varepsilon_j)=v_k(y'_k)+v_j(y'_j)$ and $y''_j q'_j + y''_k q'_k = y'_j q'_j + y'_k q'_k$. Thus, the modified solution $\my''$ satisfies the optimality conditions, but $\Psi(\my'')<\Psi(\my')$, again leading to a contradiction.
\end{proof}

\begin{lemma}\label{lem::sep-def-3}
 Let $u(\mx)=\sum_{j\in\goods} v_j(x_j)$ be a separable utility function such that $\mathbf{0} \in \FQ{u}$. Then, 
     for any $\mathbf{q} \geq \mathbf{0}$, $b>0$ and any  $\my \in \galedemand{}^u (\mathbf{0}, b)$, there exists a $\my' \in \galedemand{}^u (\mathbf{q}, b)$ such that $\my' \leq \my$. 
\end{lemma}
\begin{proof}
$\mathbf{0} \in \FQ{u}$ implies  $\sup_{\my\in\Rp^\goods} u(\my)<\infty$. Therefore, $\FQ{u}=\Rp^\goods$; in particular, $\mathbf{q} \in \FQ{u}$.
Let us pick  $\my' \in \galedemand{}^u (\mathbf{q}, b)$ such that $\sum_{j\in \goods} y'_j$ is minimal. 

Note that $u(\my)=\sup_{\my\in\Rp^\goods} u(\my)$. In particular, for any $\alpha\ge y_j$, $v_j(\alpha)=v_j(y_j)$.
For a contradiction, assume there exists a $k\in \goods$ with $y_k<y'_k$. Then, reducing $y'_k$ to $y_k$ may only increase $b\log u(\my')-\pr{\mq}{\my'}$, while decreasing $\sum_{j\in \goods} y'_j$. This contradicts the choice of $\my'$.
\end{proof}


\section{Generalized network utilities}\label{sec:network}

In this Appendix, we prove Theorem~\ref{thm:network-gale} and Theorem~\ref{thm::satiate-gale}.
We will first need a characterization of subgradients by linear programming duality. We state this in a more general form. 
\begin{restatable}{lemma}{lpsubgradient}\label{lem:lp-subgradient}
Let $A\in \R^{k\times \ell}$, $B\in \R^{k\times m}$, $\mc\in \R^\ell$, $\mb\in \R^k$. Let $\mathcal{D}\subseteq \R^m$ be the set of vectors $\mx\in\R^m$ such that the following parametric linear program is feasible and bounded, and let us define  $\varphi:\mathcal{D}\to \R$ such that $\varphi(\mx)$ is the optimum value of this program: 
\begin{equation}\tag{$\Primal(\mx)$}\label{eq:parametric-LP}
\begin{aligned}
\max~& \mc^\top \mz\\
A \mz&\le B \mx +\mb \\
\mz&\ge \0
\end{aligned}
\end{equation}
Then, for every $\mx\in\mathcal{D}$, $\mg\in\supg \varphi(\mx)$ if and only if $\mg=B^\top \bm{\pi}$ for an optimal solution $\bm{\pi}\in\Rp^k$ to the dual linear program
 \begin{equation}\tag{$\Dual(\mx)$}\label{eq:parametric-LP-dual}
\begin{aligned}
\min~& \mx^\top B^\top \bm{\pi}+\mb^\top \bm{\pi}\\
A^\top \bm{\pi}&\ge \mc\\
\bm{\pi}&\ge \0\, .
\end{aligned}
\end{equation}
\end{restatable}
\begin{proof}
First, assume that $\bm{\pi}\in\Rp^\ell$ is an optimal solution to $\Dual(\mx)$ for a given $\mx\in\R^m$. We need to show that $\mg=B^\top \bm{\pi}$ is a supergradient at $\mx$, that is, $\varphi(\mt)\le \varphi(\mx)+\pr{\mg}{\mt-\mx}$ for any $\mt\in\mathcal{D}$.

The programs $\Dual(\mx)$ and $\Dual(\mt)$ have the same feasible region, with linear objective functions $B \mx+\mb$ and $B \mt+\mb$, respectively. By weak duality, the cost of every feasible dual solution is at least the primal optimum value, thus, $\varphi(\mt)\le \mt^\top B \bm{\pi}+\mb^\top \bm{\pi}=\mg^\top \mt+\mb^\top \bm{\pi}$. The statement follows as
\[
\varphi(\mt)\le \mg^\top \mt+\mb^\top \bm{\pi}=\mg^\top (\mt-\mx)+\mx^\top B \bm{\pi} + \mb^\top \bm{\pi}=\varphi(\mx)+\pr{\mg}{\mt-\mx}\, ,
\]
noting also that $\varphi(\mx)=\mx^\top B \bm{\pi} + \mb^\top \bm{\pi}$ since $\bm{\pi}$ is an optimal solution to $\Dual(\mx)$, and $\varphi(\mx)$ is the optimum value.

Conversely, we need to show that if $\mg\in\supg \varphi(\mx)$, then there exists a dual optimal solution $\bm{\pi}$ such that $\mg=B^\top \bm{\pi}$. Let $\mz$ be a primal optimal solution to $\Primal(\mx)$. Using complementary slackness, we can turn $\Dual(\mx)$ into a feasibility system. Namely, $\mg=B^\top \bm{\pi}$ for a dual optimal solution if and only if
 \begin{equation}\label{eq:parametric-LP-dual-feas}
\begin{aligned}
B^\top \bm{\pi}&= \mg \\
A_i^\top \bm{\pi} &= c_i\, ,\quad \forall i: z_i>0\\ 
A_i^\top \bm{\pi} &\ge c_i\, ,\quad \forall i: z_i=0\\
\pi_j&=0\, , \quad \forall j: (A \mz)_j <(B \mx+\mb)_j\\
\pi_j&\ge 0 , \quad \forall j: (A \mz)_j =(B \mx+\mb)_j\, .
\end{aligned}
\end{equation}
Here, $A_i$ denotes the $i$-th column of $A$.
By Farkas's lemma, if \eqref{eq:parametric-LP-dual-feas} is not feasible, then the following system has a feasible solution:
  \begin{equation*}
\begin{aligned}
(A \mw)_j + (B \mr)_j &\le 0\, , \quad \forall j: (A \mz)_j =(B \mx+\mb)_j\\
w_i &\ge 0\, ,\quad \forall i: z_i=0\\
\mc^\top \mw + \mg^\top \mr &> 0 \, .\\ 
\end{aligned}
\end{equation*}
For a small enough $\varepsilon>0$, $\mz+\varepsilon \mw$
is feasible to the system $\Primal(\mx-\varepsilon \mr)$.
Hence, for $\mt=\mx-\varepsilon \mr$, $\varphi(\mt)\ge \mc^\top (\mz+\varepsilon \mw)=\mc^\top \mz+\varepsilon \mc^\top \mw$. Using $\mc^\top \mz=\varphi(\mx)$ by the choice of $\mz$ as an optimal solution to $\Primal(\mx)$, and the bound $\mc^\top \mw>-\mg^\top \mr$, we get 
$\varphi(\mt)> \varphi(\mx)-\varepsilon \mg^\top \mr=\varphi(\mx)+\pr{\mg}{\mt-\mx}$, a contradiction to $\mg\in \supg \varphi(\mx)$. 
 \end{proof}

Consider a generalized flow instance defined in a graph $G=(V,E)$, sink $t\in V$, capacities $\mh\in\Rp^E$, gain factors $\bm{\gamma}\in\Rpp^E$ and supply set $S=\goods$. For the rest of this section, let $u\, :\, \Rp^\goods\to \Rp$ be a generalized network utility function defined in the form \ref{eq:genflow-LP} in Section~\ref{sec:prelim}.

 Using Lemma~\ref{lem:lp-subgradient}  and complementary slackness, we get the following characterization of subgradients of generalized network utility functions. Note that for a generalized network utility function, the domain is $\mathcal{D}=\Rp^\goods$, because $\mf=0$ is always a feasible solution to \ref{eq:genflow-LP}, and the definition assumes that the objective is always bounded.
 \begin{lemma}\label{lem:genflow-subgradient}
 Let $u\, :\, \Rp^\goods\to \R_+$ be a generalized network utility function defined in the form \ref{eq:genflow-LP} for $S=\goods$. For a given supply vector $\mx\in\Rp^\goods$, let $\mf$ be an optimal solution to \ref{eq:genflow-LP}. Then, $\mg\in\supg u(\mx)$ if and only if $\mg=\bm{\pi}|_\goods$ for a vector $\bm{\pi}\in \Rp^V$ that satisfies
  \begin{equation}\label{eq:genflow-dual-opt}
\begin{aligned}
\gamma_e \pi_w -\pi_v &= 0\, , \quad \forall e=(v,w)\in E: \cp_e>f_e>0\, , \\
\gamma_e \pi_w -\pi_v &\le 0\, , \quad \forall e=(v,w)\in E: f_e=0\, , \\
\gamma_e \pi_w -\pi_v &\ge 0\, , \quad \forall e=(v,w)\in E: f_e=\cp_e\, , \\
\pi_v&=0\, ,\quad \forall v\in \goods: \net{\mf}{v}>-x_v\mbox{ or } v\in V\setminus(\goods\cup t):  \net{\mf}{v}>0\, ,\\
\pi_t&=1\, , \\
\bm{\pi}&\ge \0\, .\\
\end{aligned}
\end{equation}
 \end{lemma}

In accordance with this lemma, we say that $\mf\in \R^E$ and $\bm{\pi}\in \R^V$ form a \emph{fitting pair with respect to $\mx$}, if $\mf$ is a feasible generalized flow  with supply $\mx$, and $\bm{\pi}$ satisfies \eqref{eq:genflow-dual-opt}.
Note that this implies that $\mf$ and $\bm{\pi}$ are primal and dual optimal solutions to $\Genflow(\mx)$.
 Further, $\mg\in\supg u(\mx)$ if and only if $\mg$ is the restriction of $\bm{\pi}$ to $\goods$ for a fitting pair $(\mf,\bm{\pi})$ with respect to $\mx$. 

For every arc $e=(v,w)\in E$, let us define the \emph{reverse arc} $\rev{e}=(v,w)$ with gain factor $\gamma_{\rev{e}}=1/\gamma_e$. By sending $\alpha$ units of flow on $\rev{e}$ we mean decreasing the flow value on $e$ by $\gamma_e\alpha$.
Given a fitting pair $(\mf,\bm{\pi})$ with respect to $\mx$, we define the \emph{auxiliary network} $\Gaux_{\mf,\bm{\pi}}=(V,\Eaux_{\mf,\bm{\pi}})$ by
\[
\Eaux_{\mf,\bm{\pi}}\defeq\left\{e=(v,w)\in E\, :\, \gamma_e \pi_w -\pi_v = 0\mbox{ and }f_e<\cp_e \right\}\cup \left\{\rev{e}\, :\, e\in E\, :\, \gamma_e \pi_w -\pi_v = 0\mbox{ and }0<f_e \right\}\, .
\]
In the proof of Theorem~\ref{thm:network-gale}, we use two flow subroutines in the generalized flow network. The first one starts with a fitting pair $(\mf,\bm{\pi})$ with respect to a demand vector $\mx$, and outputs a flow $\mf'$ for a decreased demand vector $\mx'$ such that the flow value  $\net{\mf'}{t}$ is decreased to the smallest possible amount while maintaining $(\mf',\bm{\pi})$ as a fitting pair.

\begin{algorithm}[H]
   \caption{\sc{Pullback-Flow}}\label{alg:pullback}
    \KwData{$\mx\in \Rp^\goods$, and a fitting pair $(\mf,\bm{\pi})$ with respect to $\mx$}
    \KwResult{A supply vector $\0\le \mx'\le \mx$ and a feasible  generalized flow $\mf'$  such that $(\mf',\bm{\pi})$ is a fitting pair with respect to $\mx'$, and subject to this, $\net{\mf'}{t}$ is minimal.}
\end{algorithm}

\begin{lemma}\label{lem:pullback-flow}
$\Pullback$ can be formulated as a linear program. For the output flow $\mf'$, the auxiliary graph $\Eaux_{\mf',\bm{\pi}}$ does not contain any directed paths from the sink $t$ to $\supp(\mx')$.
\end{lemma}
\begin{proof}
When fixing $\bm{\pi}$,
the requirement that $(\mf',\bm{\pi})$ be a fitting pair can be formulated by linear constraints on $\mf'$. For example, for every arc $e=(v,w)$, if $\gamma_e \pi_w-\pi_v<0$ then we require $f'_e=0$, and if $\gamma_e \pi_w-\pi_v>0$ then $f'_e=\cp_e$; if $\gamma_e \pi_w-\pi_v=0$ then any value $0\le f'_e\le \cp_e$ is allowed. For the second part, assume for a contradiction that
$\Eaux_{\mf',\bm{\pi}}$ contains a directed path $P$ from $t$ to a node $j\in\goods$ with $x'_j>0$. Then, we can send a small amount of flow along this path by maintaining the flow balance on every intermediate node, and decreasing $\net{\mf'}{t}$ and $x'_j$ by positive amounts. This contradicts the minimality of  $\net{\mf'}{t}$.
\end{proof}
\begin{remark} $\Pullback$ can be in fact  formulated as a (classical) maximum flow problem. To see this, one can use  the standard generalized flow relabeling technique, first introduced by Glover~\cite{glover1973}, see also e.g., \cite{Goldberg1991,olver2020simpler}). That is, we multiply the flow at node $v$ by $\pi_v$ and replacing $\gamma_e$ by $\gamma_e \pi_w/\pi_v$ for every $e=(v,w)\in E$. On every arc, we fix $f'_e$ to the upper or lower bound unless $\gamma_e \pi_w-\pi_v=0$, in which case the relabeled gain factor is 1.
\end{remark}

Our second subroutine aims to increase the supply at a designated node $k\in\goods$ while  the supply at all other goods is non-increasing, and the flow value $\net{\mf}{t}$ is non-decreasing. The goal is again to maintain a fitting pair $(\mf',\bm{\pi})$ for a demand vector $\mx'$.

\begin{algorithm}[H]
   \caption{\sc{Reroute-Flow}}\label{alg:reroute}
    \KwData{$\mx\in \Rp^\goods$, and a fitting pair $(\mf,\bm{\pi})$ with respect to $\mx$, and a good $k\in\goods$ such that $\pi_k>0$.}
    \KwResult{A supply vector $\mx'\in\Rp^\goods$ such that $x'_k\ge x_k$ and $x'_j\le x_j$ if $j\in \goods\setminus\{k\}$, and a feasible  generalized flow $\mf'$  such that $(\mf',\bm{\pi})$ is a fitting pair with respect to $\mx'$, $\net{\mf}{t'} \geq \net{\mf}{t}$, and subject to this,
    $x'_k$ is maximal.}
\end{algorithm}

The following lemma can be derived similarly to Lemma~\ref{lem:pullback-flow}.
\begin{lemma}\label{lem:reroute-flow}
$\Reroute$ can be formulated as a linear program. For the output flow $\mf'$, the auxiliary graph $\Eaux_{\mf',\bm{\pi}}$ does not contain any directed path from $k$ to the set $(\supp(\mx')\setminus\{k\})\cup\{t\}$.
\end{lemma}

From Lemma~\ref{lem::gale-kkt}, we get the following.
\begin{lemma}\label{lem::gale-kkt-flow}
Consider a generalized network utility function $u\,:\, \Rp^\goods\to \Rp$,  $b>0$ and $\mq\in \FQ{u}$, and $\mx\in\Rp^\goods$.
Then,  $\mx\in\galedemand{}^u(\mathbf{q},b)$ if and only if there exists a fitting pair $(\mf,\bm{\pi})$ with respect to $\mx$ such that
$u(\mx)=\net{\mf}{t}$ and 
\begin{equation}\label{eq:gale-kkt-flow}
\pi_j \le \frac{q_j}{b}\cdot u(\mx)\, \quad\forall j\in \goods\, , \,  \text{ with equality for }j\in\supp(\mx)\, .  
\end{equation}

\end{lemma}
 Theorem~\ref{thm:network-gale} follows by showing the two properties in Definition~\ref{def:gale-subs} in Lemmas~\ref{lem:budget-dec} and Lemma~\ref{lem:price-dec} below. We start by property \emph{(ii)}.

\begin{lemma}\label{lem:budget-dec}
Let $u\,:\, \Rp^\goods\to \Rp$ be a generalized network utility function,  $b>b'>0$ and $\mq\in \FQ{u}$, and let $\mx\in\galedemand{}^u(\mathbf{q}, b)$. Then, there exists $\mx'\in\galedemand{}^u(\mathbf{q}, b')$ such that $\mx'\le \mx$.
\end{lemma}
\begin{proof}
The proof is given by a finite algorithm, shown as Algorithm~\ref{alg:budget-dec-alg}.
 We start from fitting pair $(\mf,\bm{\pi})$ with respect to $\mx$ as in Lemma~\ref{lem::gale-kkt-flow}. 
Thus, the inequality \eqref{eq:gale-kkt-flow} holds for $\bm{\pi}$, $\mx$, $\mq$ and $b$. If $u(\mx)=0$, then we must have $\pi_j=0$ for all $j\in\goods$ and thus \eqref{eq:gale-kkt-flow} holds also when replacing $b$ by $b'$, and therefore we can return $\mx'=\mx$. For the rest of the proof, let us assume $u(\mx)>0$. Similarly, we can assume that $\mq\neq \0$. Note that $\pi_j>0$ for every $j\in\supp(\mx)\cap \supp(\mq)$.
When replacing $b$ by $b'$, all inequalities with in \eqref{eq:gale-kkt-flow} with $q_j>0$ become strict; the ratio of the right and left hand sides will be $\alpha\defeq b/b'>1$ for all $j\in\supp(\mx)\cap \supp(\mq)$. Also, recall $u(\mx)=\net{\mf}{t}$.

Algorithm~\ref{alg:budget-dec-alg} alternates between two operations: changing $\mx$ and $\mf$ while keeping $\bm{\pi}$ unchanged; and  changing the dual $\bm{\pi}$ while keeping $\mx$ and $\mf$ unchanged. Throughout, it maintains a fitting pair with respect to a decreasing vector $\mx'$, such that  the required equalities in \eqref{eq:gale-kkt-flow} are gradually restored.

In the first step, we call $\Pullback(\mx,\mf,\bm{\pi})$ to get demands $\mx'\le \mx$ and a corresponding fitting pair $(\mf',\bm{\pi})$. We have $u(\mx')=\net{\mf'}{t}\le\net{\mf}{t}=u(\mx)$. Note that $\frac{q_j}{b'}u(\mx')/\pi_j$ is the same value $\alpha'\defeq\alpha u(\mx')/u(\mx)$ for every $j\in\supp(\mx')\cap \supp(\mq)$. 

 \paragraph{Case I:} If $\alpha'\le 1$, then we choose 
 $\lambda\defeq(\alpha-1)/(\alpha-\alpha')\in [0,1]$ such that $\lambda \alpha+(1-\lambda)\alpha'=1$. Then,  $\mx''\defeq\lambda \mx'+(1-\lambda) \mx$ and the flow $\mf''\defeq\lambda \mf'+(1-\lambda) \mf$ satisfy that $(\mf'',\bm{\pi})$ is a fitting pair with respect to $\mx''$, and \eqref{eq:gale-kkt-flow} holds  for $\bm{\pi}$, $\mx''$ and $b'$. Thus, we can terminate by outputting $\mx''\in\galedemand{}^u(\mq, b')$.

\paragraph{Case II:} If $\alpha'>1$, then we update the dual $\bm{\pi}$ to $\bm{\pi}'$ as follows. Let $T\supseteq \supp(\mx')$ be the set of nodes that can reach  $\supp(\mx')$ in the auxiliary graph $(G,\Eaux_{\mf',\bm{\pi}})$. By Lemma~\ref{lem:pullback-flow}, $t\notin T$. 
Let us define $\bm{\pi}'$ such that $\pi'_v\defeq\pi_v$ for $v\notin T$ and $\pi'_v\defeq\delta' \pi_v$ for $v\in T$ for the largest value  $1\le \delta' \le \alpha'$ such that $(\mf',\bm{\pi}')$ is a fitting pair.

We claim that $\delta'>1$. Recall that fitting pairs are required to satisfy \eqref{eq:genflow-dual-opt}.
The sign of $\gamma_e \pi'_w-\pi'_v$ is the same as the sign of $\gamma_e \pi_w-\pi_v$ for all arcs $e=(v,w)$ with both endpoints inside $T$ or outside $T$.
There is no arc $e=(v,w)$ entering or leaving $T$ with $0<f'_e<u_e$, since $e,\rev{e}\in \Eaux_{\mf',\bm{\pi}}$ for all such arcs.
For every arc $e=(v,w)$ entering $T$ with $f'_e=0$, we must have $\gamma_e \pi_w-\pi_v<0$ so that $e\notin \Eaux_{\mf',\bm{\pi}}$. Similarly, for every arc $e=(v,w)$ leaving $T$ with $f'_e=u_e$, we must have $\gamma_e \pi_w-\pi_v>0$. Hence, we can always choose $\delta'>1$.
If $\delta'=\alpha'$, then $\bm{\pi}'$, $\mx'$, $\mq$ and $b'$ satisfy \eqref{eq:gale-kkt-flow}, since the right hand side for all $j\in\supp(\mx')$ has increased by a factor $\alpha'$, and therefore equality holds for all these indices. We can therefore conclude $\mx'\in\galedemand{}^u(\mq, b)$ and terminate by outputting $\mx'$.

In case $\delta'<\alpha'$, then \eqref{eq:gale-kkt-flow} holds with strict inequality for all $j$ for $\bm{\pi}'$, $\mx'$, $\mq$, and $b'$. Moreover, the ratio of the right hand side and the left hand side is the same for all $j\in\supp(\mx')\cap \supp(\mq)$, namely, $\alpha'/\delta'<\alpha$. We now restart the above procedure with $(\mf',\bm{\pi}')$ and $\mx'$ in place of $(\mf,\bm{\pi})$ and $\mx$.

We claim that the algorithm terminates in a finite number of iterations. This holds because the linear program $\Pullback$ only depends on the sign pattern of $\gamma_e \pi_w-\pi_v$ for the arcs $e=(w,v)$ and on $\supp(\mx)$. There are $3^{|E|}$ possible sign patterns; the set $\supp(\mx)$ is monotone decreasing and therefore there are at most $|\goods|$ different possibilities.  The same configuration of arc sign patterns and $\supp(\mx)$ cannot occur twice, because 
$u(\mx')=\net{\mf'}{t}$ is strictly decreasing in every iteration of the algorithm.
 Hence, after a finite number of iterations we arrive at $\bm{\pi}'$ and $\mx'$ satisfying \eqref{eq:gale-kkt-flow} for $b'$. Since $\mx'$ is decreasing throughout, we have $\mx'\le \mx$ for the input supply vector $\mx$, as required.
\end{proof}
\begin{algorithm}[htb]
   \caption{\sc{Budget-Decrease}}\label{alg:budget-dec-alg}
    \KwData{$\mq\in \FQ{u}$, $b,b'>0$, $\mx\in \galedemand{}^u(\mq, b)$.}
    \KwResult{$\0\le \mx'\le \mx$ such that $\mx'\in \galedemand{}^u(\mq, b')$.}
    Find a fitting pair $(\mf,\bm{\pi})$ with respect to $\mx$ such that $u(\mx)=\net{\mf}{t}$ and \eqref{eq:gale-kkt-flow} holds\;
    \lIf{$u(\mx)=0$ or $\mq=\0$}{\Return{$\mx$}}
    $\bar \mx\gets \mx$ ; $\alpha\gets b/b'$ \;
    \While{$\alpha>1$}{
    $(\mx',\mf')\gets \Pullback(\bar \mx,\mf,\bm{\pi})$ \;
    $\alpha'\gets \alpha\frac{u(\mx')}{u(\mx)}$ \;
    \If{$\alpha'\le 1$}{
    $\lambda\gets\frac{\alpha-1}{\alpha-\alpha'}$ \;
    $\mx''\gets\frac{\alpha-1}{\alpha-\alpha'} \mx'+\frac{1-\alpha'}{\alpha-\alpha'} \bar \mx$ \; 
    \Return{$\mx''$}
    }\Else{
    $T\gets\left\{v\in V\, :\, v\text{ can reach }\supp(\mx')\text{ in }(G,\Eaux_{\mf',\bm{\pi}})\right\}$\, \;
    $\delta'\gets\min\left\{\alpha,\min\left\{\frac{\gamma_e\pi_w}{\pi_v}\, :\, e=(v,w)\in \din(T)\right\},\min\left\{\frac{\pi_v}{\gamma_e\pi_w}\, :\, e=(v,w)\in \dout(T)\right\}\right\}$\;
   \lIf{$\delta'=\alpha'$}{\Return{$\mx'$}}
    \lFor{$v\in T$}{$\pi_v\gets\delta'\pi_v$}
    $\alpha\gets \alpha'/\delta'$ ; $\bar \mx\gets \mx'$ \;
    }}
\end{algorithm}
We next prove property \emph{(i)} in Definition~\ref{def:gale-subs}. 

\begin{restatable}{lemma}{networkpricedec}\label{lem:price-dec}
Let $u\,:\, \Rp^\goods\to \Rp$ be a generalized network utility function and $\mq,\mq'\in \FQ{u}$ such that $\0\le\mq'\le \mq$, and let $\mx\in\galedemand{}^u(\mathbf{q}, b)$. Then, there exists $\mx'\in\galedemand{}^u(\mathbf{q}', b)$ such that $x'_j\le x_j$ whenever $q'_j=q_j$. 
\end{restatable}
\begin{proof} 
Clearly, it suffices to show that the statement holds when $q'_k<q_k$ for a single good $k\in \goods$.  Again, we start from fitting pair $(\mf,\bm{\pi})$ with respect to $\mx$ as in Lemma~\ref{lem::gale-kkt-flow}. Similarly to Lemma~\ref{lem:budget-dec}, we may assume that $u(\mx)>0$ and $\mq\neq\0$. Let us now replace $\mq$ by $\mq'$. 
The condition \eqref{eq:gale-kkt-flow} can only be violated for $j=k$.
If it is not violated, then we may return $\mx'=\mx$; this is always the case when $\pi_k=0$. For the rest of the proof, assume that the  \eqref{eq:gale-kkt-flow} is violated for $k$ with the cost $q'_k$.
Note that $\pi_j=\beta q'_j$ for every $j\neq k$, where $\beta\defeq u(\mx)/b$; we now have $\pi_k>\beta q'_k\ge0$. This value of $\beta$ will be fixed throughout the argument.

We first consider the case $q'_k>0$.
The proof uses the subroutine $\Reroute(\mx,\mf,\bm{\pi})$ with the special node $k$. This subroutine finds $\mx'$ and $\mf'$ such that $(\mf',\bm{\pi})$
 is a fitting pair with respect to $\mx'$, and further $x'_k\ge x_k$,  $x'_j\le x_j$ for $j\neq k$, and $u(\mx')=\net{\mf}{t'}\ge\net{\mf}{t}=u(\mx)$.  Consider the set of nodes $T$ reachable from $k$ in $\Eaux_{\mf',\bm{\pi}}$. According to Lemma~\ref{lem:reroute-flow}, $t\notin T$ and $T\cap \supp(\mx')=\{k\}$. Similarly as Case II in the proof of Lemma~\ref{lem:budget-dec}, we can decrease $\pi_v$ multiplicatively on $T$ to obtain $\bm{\pi}'$ such that $(\mf',\bm{\pi}')$ is a fitting pair with respect to $\mx'$. We choose the largest multiplicative decrease such that $(\mf',\bm{\pi}')$ remains a fitting pair, and that $\pi_k'\ge \beta q_k'$ is maintained.  We then repeat this procedure for $\mx',\mf'$, and $\bm{\pi}'$ in the place of $\mx,\mf$, and $\bm{\pi}$.

 We can again argue that this algorithm terminates in a finite number of iterations. Note that $\supp(\mx')$ is monotone decreasing throughout, and $\pi'_j=\beta q'_j$ is maintained for all $j\in\supp(\mx')\setminus\{k\}$. For goods with $x'_j=0$, $\pi'_j$ may decrease, and therefore $\pi'_j\le \beta q_j$ for these goods.
At termination, we have $\pi'_j=\beta q'_j$ for all $j\in \supp(\mx')$ --- including $j=k$ --- and $\pi'_j\le \beta q'_j$ whenever $x'_j=0$. Moreover, note that $x'_j\le x_j$ for all $j\neq k$.

Throughout the algorithm, $u(\mx')$ is non-decreasing. Hence, $u(\mx')/b\ge \beta$ holds. By setting $\hat b>b$ such that $u(\mx')/\hat b=\beta$, \eqref{eq:gale-kkt-flow} holds for $\mx'$, $\bm{\pi}'$, $\mq'$ and $\hat b$. Hence, $\mx'\in \galedemand{u}(\mq', \hat b)$. Applying now Lemma~\ref{lem:budget-dec} for $\mx'$, the prices $\mq'$ and budgets $\hat b\ge b$, we can find $\mx''\in \galedemand{u}(\mq',b)$ such that $\mx''\le \mx'$. Note that $x''_j\le x'_j\le x_j$ for all $j\neq k$, and therefore $\mx''$ satisfies the requirements of the Lemma.

\medskip

Consider now the case $q'_k=0$. We proceed as above, but our goal is to eventually set $\pi'_k=0$. As in the previous case, we repeatedly call $\Reroute(\mx,\mf,\bm{\pi})$ and then decrease $\pi'_v$ multiplicatively on $T$ by the largest possible factor such that $(\mf',\bm{\pi}')$ remains a fitting pair. If this largest factor is unbounded, then we can set $\pi'_v=0$ for all $v\in T$. Again, the algorithm should terminate in a finite number of iterations, using also the assumption that $\mq'\in\FQ{u}$. Once $\pi'_k=0$ is reached, we again have
$\pi'_j=\beta q'_j$ for all $j\in \supp(\mx')$, including $j=k$, and $\pi'_j\le \beta q'_j$ whenever $x'_j=0$. The rest of the proof is the same as in the case $q'_k>0$.
 \end{proof}

 This concludes the proof of Theorem~\ref{thm:network-gale}, namely, that all generalized network utilities are Gale-substitutes. Nevertheless, they are not $\Sigma$-Gale-substitutes, and so Theorem~\ref{thm::main-result} is only applicable in the case when the generalized network utilities are non-satiable. We conclude by showing  Theorem~\ref{thm::satiate-gale}. This is a modification of the argument in Section~\ref{sec:main-proof}; we omit some simple details.
 
\begin{proof}[Proof  of Theorem~\ref{thm::satiate-gale}]
In the proof of Theorem~\ref{thm::main-result}  in Section~\ref{sec:main-proof}, property \eqref{prop:satiate} is only used in the proof of Lemma~\ref{lem:gale-demands}. We show that for generalized network utilities, if $\mx\in\demand{}^u(\mp,b)$ for $\mp\in\Rp^\goods$ and further $\mx$ minimizes $\pr{\mp}{\mx}$ in $\demand{}^u(\mp,b)$, then there exists $\my\in \galedemand{}^u (\mathbf{p}, b)$ such that $\my \leq \mx$. From this property, we can derive Lemma~\ref{lem::bounds-aggregate} for thrifty competitive equilibria; the rest of the proof remains identical to the proof of Theorem~\ref{thm::main-result}.

Consider now \eqref{eq:kkt-gale} in the proof of Lemma~\ref{lem:gale-demands}. If $\lambda>0$, the same argument is applicable. Consider now the case $\lambda=0$. This means that the constraint $\pr{\mp}{\mx}\le b$ is not active in the program defining $\demand{}^u(\mp,b)$; that is, $\max_{\mz\in\Rp^\goods}\{u(\mz)\, :\, \pr{\mp}{\mz}\le b\}=\max_{\mz\in\Rp^\goods} u(\mz)=v^\star$. In particular, $u(\mx)=v^\star$. The budget minimization property means that $\mx$ is an optimal solution to the program
\begin{equation}\label{eq:genflow-LP-satiate}
\begin{aligned}
\min~& \mp^\top \mx \\
\net{\mf}{t}&\ge v^\star\\
\net{\mf}{v}&\ge -x_v\quad\forall v\in S\\
\net{\mf}{v}&\ge 0\quad\forall v\in  V\setminus (S\cup \{t\})\\
\0&\le \mf\le \mh\, .
\end{aligned}
\end{equation}
From here, using LP duality and Lemma~\ref{lem:lp-subgradient}, we can see that there exists a supergradient $\mg\in \supg u(\mx)$ and $\mu\in\Rp$ such that 
  \[ \mu g_{j} \leq p_j\, , \quad\forall j\in \goods\,,\,\mbox{with equality  if }x_{j} > 0\, .\]
If $\mu>0$, then we can use the same argument for $\lambda=1/\mu$. If $\mu=0$, given that $\mu$ acts as the dual variable of $\net{\mf}{t}\ge v^\star$, which is no longer active, the optimal value of \eqref{eq:genflow-LP-satiate} becomes $0$ since $\mathbf{0}$ is a feasible allocation of $u(\cdot)$ according to our definition. 
This implies the existence of a zero-cost optimal allocation $\mx$: $u(\mx) = v^*$ and $\mp^\top \mx = 0$, which indicates that $\mx$ belongs to both $\demand{}^u(\mp,b)$ and $\galedemand{}^u(\mp,b)$. Thus, the result follows.
\end{proof}

\section{Examples}\label{sec:egs}
In this section, we present multiple examples mentioned in the paper. 

\subsection{Connectivity of equilibria with bounded linear utilities}\label{sec:conn}

We first investigate the (dis)connectivity of competitive equilibria with bounded linear utilities. 
Our initial example illustrates that the set of thrifty competitive equilibria may be disconnected, even when agents have bounded linear utility functions, where ``thrifty" denotes that, given the prices, each agent selects the least costly bundle from the set of utility-maximizing bundles.

\begin{example}[Disconnectivity of thrifty competitive equilibrium]\label{eg:1} 
There are two agents and two goods. Each agent has a unit budget. The  agents' utility functions are
\[
\begin{aligned}
u_1(\mx_1) &= 1.3 \min \{x_{11}, 0.8\} + 0.45 x_{12}\,  ,\\
 u_2(\mx_2) &= 0.01 \min\{x_{21}, 0.3\} + 2 \min\{x_{22}, 0.8\}\, .
 \end{aligned}
 \]
First, we show that there are at least two equilibria. The first equilibrium is 
\[
p_1=1\,,\quad p_2=1\, , \quad x_{11} = 0.8\, ,\quad x_{12} = 0.2\, ,\quad x_{21} = 0.2\, ,\quad x_{22} = 0.8\, .\] Another competitive equilibrium is 
\[
p_1=1.3\,,\quad p_2=0.45\, , \quad x_{11} = 0.7\, ,\quad x_{12} = 0.2\, ,\quad x_{21} = 0.3\, ,\quad x_{22} = 0.8\, .\]
We show the second competitive equilibrium is isolated under the thrifty assumption.

First, we show the price is isolated. Suppose there is a nearby equilibrium price $\mathbf{p}' = (1.3 + \varepsilon_1, 0.45 + \varepsilon_2)$. For this price, the thrifty demand of the second agent remains the same: $x_{21} = 0.3$, $x_{22} = 0.8$. To maintain equilibrium, the first agent should receive the same allocation. As $1.3 \times 0.7 + 0.45 \times 0.2 = 1$, $\varepsilon_1$ and $\varepsilon_2$ cannot be both positive or negative. However, if  $\varepsilon_j<0$, the first agent will prefer good $j$ more at $\mathbf{p}'$. This implies the allocation of the other good of the first agent is $0$. This leads to a contradiction.

Second, we show the allocation is isolated. 
Note that the second agent is currently satiated.
Suppose there is a nearby equilibrium allocation. Agent 2 must have a different allocation, which can only be smaller. Thus,  $\mx'_2 = (x_{21} - \varepsilon_1, x_{22} - \varepsilon_2)$ with  $\varepsilon_1, \varepsilon_2\ge 0$, and at least one of them is strictly positive. In the allocation $\mx'$, both agents are non-satiated, and therefore they are spending their full budgets.  Therefore, the equilibrium prices satisfy $p'_1 + p'_2 = 2$. For the first agent, the allocation is $x'_{11} = 0.7+\varepsilon_1$, $x'_{12} = 0.2 + \varepsilon_2$. Note that both $\varepsilon_1$ and $\varepsilon_2$ are very tiny. This implies $p'_1 : p'_2 \approx 1.3 : 0.45$. Therefore, $p'_1 \approx \frac{2}{1.3 + 0.45} \times 1.3$ and $p'_2 \approx \frac{2}{1.3 + 0.45} \times 0.45$. However, the first agent cannot afford these allocations under these prices.
\end{example}

Nevertheless, it remains possible that thrifty competitive equilibria are connected through intermediate equilibria that are not thrifty. However, with a slight modification of the previous example, the next example shows that competitive equilibria are not connected without the thrifty assumption.  
\begin{example}[Disconnectivity of competitive equilibria] \label{eg:non-thrifty}
There are three agents and three goods. Each agent has a unit budget.  The utility functions are
\[
    \begin{aligned}
    u_1(\mx_1) &= 1.3 \min \{x_{11}, 0.8\} + 0.45 x_{12}\, ,\\
    u_2(\mx_2) &= 0.01 \min\{x_{21}, 0.3\} + 2 \min\{x_{22}, 0.8\} + 10^{-10} x_{23}\, ,\\
    u_3(\mx_3) &= x_{33}\, .
\end{aligned}  
\] 
First, we show there are at least two equilibria.
The first equilibrium is with $\mathbf{p} = (1, 1, 1)$, and the allocation is $\mx_1 = (0.8, 0.2, 0)$, $\mx_2 = (0.2, 0.8, 0)$, and $\mx_3 = (0, 0, 1)$.
The second equilibrium is with $\mathbf{p} = (1.3, 0.45, 1.25)$ and the allocation is $\mx_1 = (0.7, 0.2, 0)$, $\mx_2 = (0.3, 0.8, 0.2)$, and $\mx_3  = (0, 0, 0.8)$.

In order to show our results, we first state the following properties: for any equilibrium, 
\begin{enumerate}
    \item $p_3 > 0$ as the third agent is only interested in the third good;
    \item  $x_{13} = 0$ as the first agent shows no interest in the third good and its price is strictly positive;
    \item $p_1 > 0$ and $p_2 > 0$ as, otherwise, these goods will be over-demanded;
    \item $x_{31} = x_{32} = 0$ as the third agent shows no interest in the first two goods and their prices are strictly positive;
    \item $p_1 + p_2 + p_3 = 3$ as all agents are non-satiated.
\end{enumerate}

Now, we show the second equilibrium is isolated.

First, we demonstrate that the price is isolated. Suppose there is a nearby equilibrium price $\mathbf{p}' = (1.3 + \varepsilon_1, 0.45 + \varepsilon_2, 1.25 + \varepsilon_3)$. For this price, the allocation of the second agent on the first two goods is the same. Therefore, the first agent should receive the same allocation. As $1.3 \times 0.7 + 0.45 \times 0.2 = 1$, $\varepsilon_1$ and $\varepsilon_2$ cannot be both positive or negative. Similar to the previous example, if one of $\varepsilon_1$ and $\varepsilon_2$ is negative, then the first agent will prefer this good more at this nearby price. This implies the allocation of the other good of the first agent is $0$. This leads to the conclusion $\varepsilon_1 = \varepsilon_2 = \varepsilon_3 = 0$.

Second, we show the allocation is isolated. Suppose there is a nearby equilibrium allocation $\mathbf{x}'$ with the equilibrium $\mathbf{p}'$. As $x_{11} < 0.8$, $p'_1 : p'_2 = 1.3 : 0.45$.\footnote{Otherwise, the allocation of the first agent will be changed significantly.} Additionally, $p'_1 \times x'_{11} + p'_2 \times x'_{12} = 1$. This implies $p_1'$ is very close to $p_1$ and $p_2'$ is very close to $p_2$: if $\| \mx' - \mx\|_{\infty} \leq \varepsilon$, then $|p_1' - p_1| \leq O(\varepsilon)$ and $|p_2' - p_2| \leq O(\varepsilon)$.\footnote{Note that $p'_1 : p'_2 = 1.3 : 0.45 = p_1 : p_2$. By budget constraint, $\frac{1.3}{0.45} p'_2 x'_{11} + p'_2 x'_{12} = \frac{1.3}{0.45}p_2 x_{11} + p_2 x_{12}$. Therefore, $p_2' - p_2 = \frac{p'_2 (\frac{1.3}{0.45} (x_{11} - x'_{11}) + (x_{12} - x'_{12}))}{\frac{1.3}{0.45} x_{11} + x_{12}}$. We have a similar result for $p_1' - p_1$.} This also implies $|p_3' - p_3| \leq O(\varepsilon)$ as $p'_1 + p'_2 + p'_3 = 3$. Since price is isolated, $p'_1 = p_1$,  $p'_2 = p_2$, and $p'_3 = p_3$. Under this price, $\mathbf{x}$ is the only allocation.
\end{example}

\subsection{Non-monotonicity of Gale Demands}\label{sec:ex}
The next example shows the somewhat counterintuitive property that increasing the prices may lead to an increase in the Gale utility. 
\begin{example}[Non-monotonicity of Gale demands]\label{eg:mono}
Consider a market with two goods and a single agent with a unit budget, $b = 1$. The agent's utility function is 
$$u(\mx) = \min\{ x_{1} + x_{2}, 2 + 0.1 \cdot x_{1} + 0.2 \cdot x_{2}\}.$$ 

Consider the price vector: 
$$\mathbf{q} = \left(\frac{0.1}{2 / 0.9 + 0.2 \cdot 0.01}, \frac{0.2}{2 / 0.9 + 0.2 \cdot 0.01}\right).$$
At this price, the agent's Gale demand and corresponding utility are
$$\galedemand{}(\mathbf{q},b) = \{(2 / 0.9, 0.01)\}, \quad u(\galedemand{}(\mathbf{q},b)) = 2 / 0.9 + 0.2 \cdot 0.01.$$ 

Now, consider a higher price vector: $$\mathbf{q}' = \left(\frac{1}{0.01  + 2 / 0.85}, \frac{1}{0.01 + 2 / 0.85}\right).$$ 
Under price $\mathbf{q}'$, the agent's Gale demand and utility become
$$\galedemand{}(\mathbf{q}',b) = \{(0.01, 2 / 0.85)\} \quad u(\galedemand{}(\mathbf{q}',b)) = 2 / 0.85 + 0.01.$$ 

Clearly, $\mathbf{q} < \mathbf{q}'$ component-wise, yet we observe that $u(\galedemand{}(\mathbf{q},b)) <  u(\galedemand{}(\mathbf{q}', b))$.
\end{example}
It is also worth noting that such a case cannot arise under separable utilities (see the proof of Lemma~\ref{lem::sep-def-1} for further details).

\subsection{(In)-approximate Nash Welfare utility by Equilibrium utility}
Theorem~\ref{thm::main-result}  shows that  for $\Sigma$-Gale-substitutes utilities, $u_i(\my_i)\ge \frac{1}{2}u_i(\mx_i)$ for the Nash welfare maximizing allocation $(\my_i)_{i\in \agents}$ and CE allocation $(\mx_i)_{i\in \agents}$. One may ask the converse: namely,
for any agent $i$, is there a CE $(\mx_i)_{i\in \agents}$ such that $u_i(\mx_i)\ge  C \cdot u_i(\my_i)$ for some constant $C>0$\footnote{We note that the utility achieved by an agent in a Nash welfare maximizing allocation is unique, since the logarithm of the Nash welfare is a strictly concave function.}?  We show a constant approximation is not achievable by providing an example in which the CE is unique and there exists an agent whose utility at this CE is much worse than his/her utility at Nash welfare maximizing allocation. 

\begin{example}\label{example::E4}
Consider a market with $n + 2$ agents and $n + 1$ goods. Let $\varepsilon < 1$ be a small constant. For each agent $1 \leq i \leq n$, the utility function is given by:
\[
u_i(\mx_i) = \min\{x_{ii}, 0.5\} + \varepsilon x_{i(n+1)}.
\]
For the remaining two agents, define their utilities as:
\[
u_{n+1}(\mx_{n+1}) = x_{(n+1)(n+1)}, \quad \text{and} \quad u_{n+2}(\mx_{n+2}) = \sum_{j = 1}^{n} x_{(n+2)j}\, .
\]

We claim that the unique competitive equilibrium occurs with prices:
\[
p_j = \frac{2}{n} \quad \text{for } 1 \leq j \leq n, \qquad p_{n+1} = n.
\]
At this equilibrium, the allocation is as follows: 
\begin{itemize}
    \item for agent $i$ such that $1 \leq i \leq n$:
$x_{ii} = 0.5$ and $x_{i(n+1)} = \frac{n - 1}{n^2}$;
    \item for agent $n+1$: $x_{(n+1)(n+1)} = \frac{1}{n}$;
    \item for agent $n+2$: $x_{(n+2)j} = 0.5$ for  $1 \leq j \leq n$.
\end{itemize}

We now show that this is the only competitive equilibrium.

First, for all $1 \leq j \leq n$, $p_j > 0$, since otherwise agent $n+2$ could obtain an unbounded amount of good $j$, violating feasibility. Given $p_j > 0$, and since $x_{(n+2)j} \geq 0.5$ for all $1 \leq j \leq n$, the first $n$ goods share the same prices, and, additionally,  the total spending of agent $n+2$ satisfies: $ 1 = \sum_{j = 1}^n x_{(n+2)j} p_j \geq  0.5 \cdot p_j \cdot n$, which implies $p_j \leq \frac{2}{n}$. We now show $p_j = \frac{2}{n}$. If $p_j < \frac{2}{n}$, then, given $\sum_j p_j = \sum_i b_i$,  $p_{n+1} = \sum_i b_i - \sum_{j=1}^n p_j > n$. In this case, agents $1 \leq i \leq n$ receive $x_{ii} = 0.5$, and therefore $x_{(n+2)i} = 1 - 0.5 = 0.5$ for all $1 \leq i \leq n$. This implies $0.5 \cdot p_j \cdot n = 1$, which leads back to $p_j = \frac{2}{n}$.

Therefore, this pricing and allocation define the unique competitive equilibrium.

Finally, observe that agent $n+1$ receives utility $\frac{1}{n}$ in this equilibrium. In contrast, under the allocation that maximizes Nash welfare, for sufficiently small $\varepsilon$, agent $n+1$ would receive utility $1$. This illustrates that the utility an individual agent receives in a Nash welfare maximizing allocation can be substantially higher than in the corresponding market equilibrium.
\end{example}

\subsection{(In)-approximation for general utilities}
Our next example demonstrates if the utility function is not Gale substitutes, then agents' competitive equilibrium utilities can be much higher than the Nash welfare maximizing solution.
\begin{example}\label{example:non-gale}
     We consider an economy with three types of agents (A, B, C): 
     \[ \text{$1$ type-A agent,\quad $n^2$ type-B agents,\quad $n$ type-C agents;}\] 
     and four types of goods (O, P, R, S): 
     \[ \text{$1$ type-O good,\quad $n$ type-P goods,\quad $n^2$ type-R goods,\quad $n$ type-S goods.}\] In total, there are $N = n^2 + n + 1$ agents and $M = n^2 + 2n + 1$ goods.

The \textbf{type-A agent} is interested in the type-O good and type-P goods with equal preferences: $v_{ij} = 1$ for $i = A$ and $j \in O \cup P$. However, the utility from each good is capped at $1$, so the utility is given by: $$u_{A}(\mx_A) = \sum_{j \in O \cup P} \min \{x_{Aj}, 1\}.$$

The \textbf{type-B agents} are indexed by pairs $(s, t)$ where $s, t \in [n]$, for a total of $n^2$ agents. Similarly, the type-R goods are indexed by $(s, t)$.  Each type-B agent $(s, t)$ is interested in 
\begin{itemize}
    \item the $s$-th type-P good, and
    \item the type-R good with index $(s,t)$.
\end{itemize}
The utility function of type-B agent with index $i = (s, t)$ is $$u_i(\mx_i) = \min \{ x_{i P_s} + \frac{n}{n+1}, x_{i R_{s, t}}\},$$ 
where $x_{i P_s}$ is agent $i$'s allocation of the $s$-th type-P good and $x_{i R_{s, t}}$ is agent $i$'s allocation of the type-R good with index $(s,t)$.

The \textbf{type-C agents} are indexed by $s \in [n]$. Each type-C agent with index $s$ is only interested in the $s$-th type-S good and all type-R goods with index $(s, \cdot)$, and the utility from each good is capped at $1$. The utility function type-C agent with index $i = s$ is 
$$u_i(\mx_i) = \sum_{j \in {S_s} \cup {R_{(s, \cdot)}}} \min\{x_{ij}, 1\},$$
where $S_s$ represents the $s$-th type-S good and $R_{(s, \cdot)}$ represents the set of type-R goods with index $(s, \cdot)$.

\paragraph{The maximal Nash Welfare solution}
Let $k$ be a parameter, which will be specified later.

The \textbf{type-A agent} receives the type-O good and $k + \frac{1}{n+1}$ portion of type-P goods. This results in a total utility of $1 + nk + \frac{n}{n+1}$. Each \textbf{type-B agent} obtains $ 1 - \frac{ k}{n}$ portion of corresponding type-R good and $\frac{1}{n+1} - \frac{k}{n}$ portion of corresponding type-P good, yielding a utility of $1 - \frac{k}{n}$. The \textbf{type-C agents} receive the remaining goods, resulting in a utility of $1 + k$.

The goal is to maximize $(1 + nk + \frac{n}{n+1}) (1 - \frac{k}{n})^{n^2} (1 + k)^n$. By calculation, the optimal $k$ is at most $n^{-\frac{1}{2}}$.\footnote{First, we take the log and calculate the gradient, which is $\frac{n}{1 + nk + \frac{n}{n+1}} + \frac{n}{1 + k} - \frac{n}{1 - \frac{k}{n}}$. Note that this is a decreasing function with $k$. So, in order to show the optimal $k$ is no larger than $n^{-\frac{1}{2}}$, we only need to show when $k =n^{-\frac{1}{2}}$, $\frac{n}{1 + nk + \frac{n}{n+1}} + \frac{n}{1 + k} \leq \frac{n}{1 - \frac{k}{n}}$. This is true as $\frac{n}{1 - \frac{k}{n}} - \frac{n}{1 + k} =  \frac{(n+1) k}{1 + k - \frac{k}{n} - \frac{k^2}{n}} \geq nk = \frac{1}{k} \geq \frac{n}{1 + nk + \frac{n}{n+1}}$.} Therefore, the type-A agent receives at most $1 + \sqrt{n} + \frac{n}{n+1}$ utility in the maximal Nash welfare solution.

\paragraph{Thrifty competitive equilibrium}
At the thrifty competitive equilibrium, the prices of type-O, P, and S goods are $0$, while the prices of type-R goods are $ \frac{n+1}{n}$. In this case, the \textbf{type-A agent} will get the type-O good and type-P goods, resulting in a utility of $n+1$.
Each \textbf{type-B agent} receives $\frac{n}{n+1}$ portion of the corresponding type-R good, leading to a utility of $\frac{n}{n+1}$.
The \textbf{type-C agent} acquires $1$ portion of the corresponding type-S good and the rest of type-R goods, resulting in  a utility of $1 + \frac{n}{n+1}$.

\vspace{0.2in}
Comparing the type-A agent's utility in the two scenarios, the type-A agent achieves at most $1 + \sqrt{n} + \frac{n}{n+1}$ utility in the maximal Nash welfare solution, whereas the agent attains $n+1$ utility in the (thrifty) competitive equilibrium.
\end{example}

\subsection{Example with Leontief-free utility functions}\label{exp:leontief-free}
Not all Leontief-free utility functions satisfy property \eqref{prop:satiate}. The subsequent example illustrates a scenario where an agent can attain significantly more utility in competitive equilibrium compared to the Nash welfare maximizing allocation when agents have Leontief-free utilities.
\begin{example} \label{exp:leontief-free-satiated}
    Consider the scenario with $n$ goods and $n+1$ agents. Agent $i = 1, 2, \cdots, n$ only likes good $i$, has budget $b_i = 1$ and a linear utility function $u_i(\mx_i) = x_{ii}$. Agent $n+1$ has a substantially larger budget $b_{n+1} = B$, and a utility function $\min\{\sum_{j=1}^n x_{(n+1) j}, (n-1)(1 - \varepsilon)\}$ with a small $\varepsilon > 0$.

    There is a competitive equilibrium where the prices are 
    $$p_1 = 1 \text{ and } p_2 = p_3 = \cdot = p_n = \frac{1}{\varepsilon}\, .$$ 
    Agent $1$ obtains $x_{11} = 1$, the other agents get $x_{ii} = \varepsilon$ for $i = 2, \cdots, n$, and agent $n+1$ gets $x_{(n+1)1} = 0$ and $x_{(n+1)j} = 1 - \varepsilon$ for $j = 2, \cdots, n$. The competitive equilibrium holds if $B > \frac{(n-1)(1 - \varepsilon)}{\varepsilon}$.

    On the other hand, in the Nash welfare maximizing allocation, when $B$ is big enough, agent $n+1$ receives (approximately) $\frac{(1 - \varepsilon)(n-1)}{n}$ of each good, with the remaining goods shared among the remaining agents. Agent $1$ receives a utility of $\varepsilon + \frac{1 - \varepsilon}{n}$, which is significantly less than the utility received in the competitive equilibrium.

    In this case, agent $n+1$'s utility doesn't satisfy property \eqref{prop:satiate}. When $\mp = \mathbf{0}$, $$\mx_{n+1} = (0, 1 - \varepsilon, 1 - \varepsilon, \cdots, 1 - \varepsilon) \in \galedemand{n+1}(\mp, B)\, .$$  
    However, when $\mp = (1, \frac{1}{\varepsilon}, \frac{1}{\varepsilon}, \cdots, \frac{1}{\varepsilon})$, $$\text{for any } \mx_{n+1} \in \galedemand{n+1}(\mp, B), \ \ x_{(n+1)1} > 0,$$ which violates property \eqref{prop:satiate}.
    
Observe that we can make the budgets of all agents the same by creating $B$ copies of agent $n+1$, each with the utility function $\min\{\sum_{j=1}^n x_{(n+1) j}, (n-1)(1 - \varepsilon)/B\}$.
\end{example}

\subsection{Example with matching utility functions}
The following example shows that matching utility functions (i.e., those with a unit-demand constraint) are not Gale-substitutes in general. This also implies that utility functions with more general constraints (as in~\cite{garg2022approximating,jalota2023fisher}) need not satisfy the Gale-substitutes property either. 
\begin{example} \label{exp:matching}
  Consider a market with $3$ goods and an agent with unit budget. The agent has a \emph{matching utility function} (i.e., a unit-demand constraint), with valuations of $1$, $2.5$, and $0.7$ for goods 1, 2, and 3, respectively. The initial prices vector is $(0.3, 1.2, 0.1)$. In this case, the Gale demand is 
   $${\arg\max}_{\mx : x_1 + x_2 + x_3 \leq 1} \{\log (x_1 + 2.5 x_2 + 0.7 x_3) - (0.3 x_1 + 1.2 x_2 + 0.1 x_3)\}, $$
   which yields the solution $(0, 103/198, 95/198)$. Now, keeping the prices of the first two goods unchanged, we increase the price of the third good from $0.1$ to $0.6$. Under these new prices, the Gale demand becomes $(5/9, 4/9, 0)$. Since the demand for good $2$ decreases as the price of good $3$ increases, this shows that the utility function does not satisfy the Gale-substitute property.
\end{example}

\bibliography{sample}

\end{document}